\documentclass[11pt]{article}

\usepackage[letterpaper,margin=1.00in]{geometry}
\usepackage{amsmath, amssymb, amsthm, amsfonts}
\usepackage{authblk}

\newcommand{\para}[1]{\paragraph{#1}}
\newtheorem{theorem}{Theorem}[section]
\newtheorem{lemma}[theorem]{Lemma}

\newtheorem{corollary}[theorem]{Corollary}

\newtheorem{definition}{Definition}[section]

\usepackage{microtype}
\usepackage{cite}
\usepackage{framed}
\usepackage[framemethod=tikz]{mdframed}
\usepackage{appendix}
\usepackage{graphicx}
\usepackage{color}
\usepackage{varwidth}
\usepackage{wrapfig}
\usepackage[boxed,noend,linesnumbered]{algorithm2e}
\usepackage[noend]{algpseudocode}
\usepackage[textsize=tiny]{todonotes}
\usepackage{array}
\usepackage[normalem]{ulem}
\usepackage{dsfont}

\usepackage[shortlabels]{enumitem}
\setitemize{noitemsep,topsep=3pt,parsep=3pt,partopsep=3pt}

\usepackage{xspace}
\usepackage{xifthen}

\definecolor{darkgreen}{rgb}{0,0.5,0}
\definecolor{darkblue}{rgb}{0,0,0.8}
\usepackage{hyperref}
\hypersetup{
    unicode=false,          
    colorlinks=true,        
    linkcolor=darkblue,          
    citecolor=darkgreen,        
    filecolor=magenta,      
    urlcolor=cyan           
}
\RequirePackage[]{silence}
\usepackage[capitalize, nameinlink]{cleveref}

\renewcommand{\vec}[1]{\ensuremath{\boldsymbol{#1}}}

\newcommand{\calG}{\ensuremath{\mathcal{G}}}

\newcommand{\calP}{\ensuremath{\mathcal{P}}}

\newcommand{\wmin}{\ensuremath{w_{\mathrm{min}}}}
\newcommand{\wmax}{\ensuremath{w_{\mathrm{max}}}}

\newcommand{\ignore}[1]{}

\algnewcommand\algorithmicswitch{\textbf{switch}}
\algnewcommand\algorithmiccase{\textbf{case}}

\algdef{SE}[SWITCH]{Switch}{EndSwitch}[1]{\algorithmicswitch\ #1\ \algorithmicdo}{\algorithmicend\ \algorithmicswitch}%
\algdef{SE}[CASE]{Case}{EndCase}[1]{\algorithmiccase\ #1}{\algorithmicend\ \algorithmiccase}%
\algtext*{EndSwitch}%
\algtext*{EndCase}%

\newcommand{\CONGEST}{\ensuremath{\mathsf{CONGEST}}\xspace}
\newcommand{\SUPPORTED}{\ensuremath{\mathsf{SUPPORTED}\ \mathsf{CONGEST}}\xspace}
\newcommand{\LOCAL}{\ensuremath{\mathsf{LOCAL}}\xspace}

\newcommand{\eps}{\varepsilon}
\renewcommand{\epsilon}{\varepsilon}
\newcommand{\poly}{\operatorname{\text{{\rm poly}}}}

\newcommand{\set}[1]{\left\{#1\right\}}

\DeclareMathOperator{\E}{\mathbb{E}}

\newcommand{\hide}[1]{}

\renewcommand{\phi}{\varphi}

\renewcommand{\Pr}{\mathbb{P}}

\title{Distributed CONGEST Approximation of\\ Weighted Vertex Covers and
  Matchings} 

\author{Salwa Faour\thanks{salwa.faour@cs.uni-freiburg.de}}
\author{Marc Fuchs\thanks{marc.fuchs@cs.uni-freiburg.de}}
\author{Fabian Kuhn\thanks{kuhn@cs.uni-freiburg.de}}
\affil{University of Freiburg, Germany}
\date{}

\begin{document}

\maketitle

\begin{abstract}
  We provide \CONGEST model algorithms for approximating the
  minimum weighted vertex cover and the maximum weighted matching
  problem. For bipartite graphs, we show that a $(1+\eps)$-approximate
  weighted vertex cover can be computed deterministically in
  $\poly\big(\frac{\log n}{\eps}\big)$ rounds. This generalizes a
  corresponding result for the unweighted vertex cover problem shown
  in [Faour, Kuhn; OPODIS '20]. Moreover, we show that in general
  weighted graph families that are closed under taking subgraphs and
  in which we can compute an independent set of weight at least
  $\lambda\cdot w(V)$ (where $w(V)$ denotes the total weight of all
  nodes) in polylogarithmic time in the \CONGEST model, one can
  compute a $(2-2\lambda +\eps)$-approximate weighted vertex cover in
  $\poly\big(\frac{\log n}{\eps}\big)$ rounds in the \CONGEST model.
  Our result in particular implies that in graphs of arboricity $a$,
  one can compute a $(2-1/a+\eps)$-approximate weighted vertex cover
  problem in $\poly\big(\frac{\log n}{\eps}\big)$ rounds in the
  \CONGEST model.

  For maximum weighted matchings, we show that a
  $(1-\eps)$-approximate solution can be computed deterministically in
  time $2^{O(1/\eps)}\cdot \poly\log n$ in the \CONGEST model. We also
  provide a randomized algorithm that with arbitrarily good constant
  probability succeeds in computing a $(1-\eps)$-approximate weighted
  matching in time
  $2^{O(1/\eps)}\cdot \poly\log(\Delta W)\cdot \log^* n$, where $W$
  denotes the ratio between the largest and the smallest edge
  weight. Our algorithm generalizes results of [Lotker, Patt-Shamir,
  Pettie; SPAA '08] and [Bar-Yehuda, Hillel, Ghaffari, Schwartzman;
  PODC '17], who gave $2^{O(1/\eps)}\cdot \log n$ and
  $2^{O(1/\eps)}\cdot \frac{\log\Delta}{\log\log\Delta}$-round
  randomized approximations for the unweighted 
  matching problem.

  Finally, we show that even in the \LOCAL model and in bipartite
  graphs of degree $\leq 3$, if $\eps<\eps_0$ for some constant
  $\eps_0>0$, then computing a $(1+\eps)$-approximation for the
  unweighted minimum vertex cover problem requires
  $\Omega\big(\frac{\log n}{\eps}\big)$ rounds. This generalizes a
  result of [G\"o\"os, Suomela; DISC '12], who showed that computing a
  $(1+\eps_0)$-approximation in such graphs requires $\Omega(\log n)$
  rounds.
\end{abstract}

\section{Introduction and Related Work}
\label{sec:intro}

Maximum matching (MM) and minimum vertex cover (MVC) 
are two classic optimization problems that have been studied
intensively in the context of distributed graph algorithms~(e.g., \cite{ahmadi18,ahmadi20,assadi2019coresets,AstrandFPRSU09,AstrandS10,Ben-BasatEKS18,yehuda16,yehuda17,Bar-YehudaCMPP20,disc19_optcovering,czygrinow03,czygrinow2004,even15,opodis20_MVC,rounding,GhaffariJN20,goeoes14_DISTCOMP,GrandoniKP08,GrandoniKPS08,harris20_matching,hoepman06,itai86,KoufogiannakisY11,lowerbound,nearsighted,lotker09,lotker15,nieberg08,wattenhofer04}). The
problems are closely related to each other: the fractional relaxations
of the unweighted variants of the problem are linear programming (LP)
duals of each other. The problems are however also fundamentally
different. While a maximum (weighted) matching can be found in
polynomial time in all graphs~\cite{edmonds65a,edmonds65b}, for the minimum
(weighted) vertex cover problem, this is only true for bipartite
graphs~\cite{egervary31,koenig31}. In general graphs, even for the unweighted
MVC problem, the best polynomial-time approximation algorithms have an
approximation ratio of $2-o(1)$~\cite{Karakostas09}. The MVC problem
is known to be APX-hard~\cite{dinur05,hastad01}, and if the unique games conjecture
holds, the current $(2-o(1))$-approximation algorithms are essentially
best possible~\cite{MVC_UGChard}.

In the distributed context, most prominently, the problems have been studied in the standard
message passing models in graphs, in the \LOCAL model and the \CONGEST
model~\cite{peleg00}. In both models, the graph $G=(V,E)$ on which we
want to solve some graph problem also represents the network and it is
assumed that the nodes $V$ of $G$ can communicate with each other in
synchronous rounds by exchanging messages over the edges $E$ of
$G$. In the \LOCAL model, the size of those messages is not
restricted, whereas in the \CONGEST model, it is assumed that each
message has to consist of at most $O(\log n)$ bits, where $n=|V|$ is
the number of nodes of the network graph $G$. In the following
discussion of existing work on distributed matching and
vertex cover algorithms, we concentrate on
polylogarithmic-time distributed algorithms that also work for the
weighted variants of the problems.

\para{Distributed Complexity of Weighted Matchings.} While the
unweighted versions of both problems can be approximated within a
factor of $2$ by computing a maximal matching, a little more work is
needed for weighted matchings and vertex covers. The first
polylogarithmic-time distributed algorithm for computing a constant
approximation for the maximum weighted matching (MWM) problem was
presented in \cite{wattenhofer04}. This algorithm was then improved in
\cite{lotker09} and in \cite{lotker15}, where it is shown that a
$(1/2-\eps)$-approximation for MWM can be computed in
$O(\log(1/\eps)\log n)$ rounds in the \CONGEST model. In
\cite{yehuda17}, it was further shown that can one compute a
$1/2$-approximation for MWM in time
$O\big(\log W\cdot T_{\mathrm{MIS}}\big)$ in the \CONGEST model, where
$W$ is the ratio between the largest and smallest edge weight and
where $T_{\mathrm{MIS}}$ is the time for computing a maximal
independent set. The paper also shows that for constant $\eps$, a
$(1/2-\eps)$-approximation can be computed in only
$O\big(\frac{\log\Delta}{\log\log\Delta}\big)$ rounds. Note that as
shown in \cite{nearsighted}, this time complexity is best possible for
any constant approximation algorithm, even in the \LOCAL model. All
the above algorithms are randomized. In \cite{rounding}, Fischer gave
a deterministic \CONGEST algorithm to compute a
$(1/2-\eps)$-approximate weighted matching with a round complexity of
$O(\log^2\Delta\cdot \log1/\eps + \log^* n)$. This algorithm was
refined in \cite{ahmadi18}, where it was shown that in time
$O\big(\frac{\log^2 (\Delta/\eps)+\log^* n}{\eps}+\frac{\log(\Delta
  W)}{\eps^2} \big)$, it is even possible to deterministically compute
a $(2/3-\eps)$-approximation for the MWM problem in general graphs and
a $(1-\eps)$-approximation for the MWM problem in bipartite graphs, in
the \CONGEST model. To the best of our knowledge, this is the only
existing polylogarithmic-time \CONGEST algorithm to obtain an
approximation ratio that is better than $1/2$. It has been observed
already in \cite{KoufogiannakisY11,lotker15,nieberg08} that in the
\LOCAL model, better approximations for maximum weighted matching can
be computed efficiently. In particular, \cite{lotker15,nieberg08} show
that even in general graphs, a $(1-\eps)$-approximation can be
computed in $\poly\big(\frac{\log n}{\eps}\big)$ rounds. It has later
been shown that this can also be achieved
deterministically~\cite{stoc18_edgecoloring}. The best known \LOCAL
MWM approximation algorithms are by Harris~\cite{harris20_matching},
who shows that a $(1-\eps)$-approximation can be computed in
randomized time
$\tilde{O}\big(\frac{\log\Delta}{\eps^3}\big) +
\poly\log\big(\frac{\log\log n}{\eps}\big)$ and in deterministic time
$\tilde{O}\big(\frac{\log^2\Delta}{\eps^4}+\frac{\log^*
  n}{\eps}\big)$. Those algorithms are based on computing large
matchings in hypergraphs defined by paths of length $O(1/\eps)$ and
they unfortunately cannot directly turned into efficient \CONGEST
algorithms. To the best of our knowledge, even constant $\eps>0$, the
only efficient \CONGEST algorithms are for the unweighted maximum
matching problem. Lotker, Patt-Shamir, and Pettie~\cite{lotker15} give
an algorithm to compute a $(1-\eps)$-approximation for the
\emph{unweighted} maximum matching problem in time only
$2^{O(1/\eps)}\cdot \log n$ in the randomized \CONGEST model. In
\cite{yehuda17} (full version), this was even improved to
$2^{O(1/\eps)}\cdot\frac{\log\Delta}{\log\log\Delta}$. As one of our
main contributions, we obtain similar algorithms for the weighted
matching problem.  Obtaining a $(1-\eps)$-approximation in
$\poly\big(\frac{\log n}{\eps}\big)$ \CONGEST rounds is one of the key
open questions in understanding the distributed complexity of maximum
matching. Fischer, Mitrovi\'{c}, and Uitto
~\cite{streaming_polymatching} recently settled a related problem for
unweighted matchings in the streaming model and in the latest
versionof their paper, they even obtain a
$\poly\big(\frac{\log n}{\eps}\big)$-round \CONGEST algorithm for the
unweighted matching problem.

\para{Distributed Complexity of Weighted Vertex Covers.} The first
distributed constant-factor approximation algorithm for the minimum
weighted vertex cover (MWVC) problem is due to Khuller, Vishkin, and
Young~\cite{KhullerVY94}. They describe a simple deterministic
algorithm to obtain a $(2+\eps)$-approximation for MWVC. The algorithm
can directly be implemented in $O(\log(n)\cdot\log(1/\eps))$ rounds in
the \CONGEST model. The time for computing a $(2+\eps)$-approximation
has subsequently been improved to $O(\log(\Delta)/ \poly(\eps))$ in
\cite{nearsighted} and to $O(\log\Delta/\log\log\Delta)$ in
\cite{yehuda16,Ben-BasatEKS18,disc19_optcovering} (with a very minor
dependency on $\eps$ in \cite{disc19_optcovering}). Note that as for
maximum matching, this dependency on $\Delta$ is optimal for any
constant-factor approximations~\cite{lowerbound}. The algorithm of
\cite{disc19_optcovering} can also be used to compute a
$2$-approximate weighted vertex cover in time $O(\log n)$. Other
polylogarithmic-time algorithms to compute $2$-approximations for MWVC
appeared in \cite{GrandoniKPS08,KhullerVY94,KoufogiannakisY11}. In the
\LOCAL model, one can use generic techniques from
\cite{ghaffari2017complexity,polylogdecomp} (or the techniques from
this paper) to deterministically compute a $(1+\eps)$-approximate
minimum weighted vertex cover in time
$\poly\big(\frac{\log n}{\eps}\big)$. Further, in
\cite{goeoes14_DISTCOMP}, it was shown that even on bipartite graphs
with maximum degree $3$, there exists a constant $\eps_0>0$ such that
computing a $(1+\eps_0)$-approximate (unweighted) vertex cover
requires $\Omega(\log n)$ rounds, even in the \LOCAL model and even
when using randomization. We generalize this result and show that for
computing a $(1+\eps)$-approximation, one requires
$\Omega(\log(n)/\eps)$ rounds. While for maximum matching, there are
several \CONGEST algorithms that achieve approximation ratios that are
better than $1/2$, for the minimum vertex cover problem, efficiently
achieving an approximation ratio significantly below $2$ in general
graphs might be a hard problem. In this case, computing an exact
solution even has a lower bound of $\tilde{\Omega}(n^2)$ rounds in the
\CONGEST model and it is therefore basically as hard as any graph
problem can be in this model~\cite{censorhillel_disc17}. To what
extent we can achieve approximation ratios below $2$ in the \CONGEST
model for variants of the minimum vertex cover problem is an
interesting open question. There recently has been some progress. In
\cite{benbasat_param}, it is shown that the minimum vertex cover
problem (and also the maximum matching problem) can be solved more
efficiently if the optimal solution is small. In particular, if the
size of an optimal vertex cover is at most $k$, a minimum vertex cover
can be computed deterministically in time $O(k^2)$ and a
$(2-\eps)$-approximate solution can be computed deterministically in
time $O(k + (\eps k)^2)$ (and slightly more efficiently with
randomization). This was the first efficient \CONGEST algorithm that
achieves an approximation ratio below $2$ for the minimum vertex cover
problem for some graphs. In \cite{opodis20_MVC}, it was shown that in
bipartite graphs, a $(1+\eps)$-approximation can be computed in time
$\poly\big(\frac{\log n}{\eps}\big)$. One of the main results of this
paper is a generalization of this result to the weighted vertex cover
problem. Further, it has recently been shown that on the square graph
$G^2$, it is possible to compute a $(1+\eps)$-approximate (unweighted)
vertex cover in time $O(n/\eps)$ in the \CONGEST model (on
$G$)~\cite{Bar-YehudaCMPP20}.
 
 \subsection{Our Contributions}
 \label{sec:contribution}

 We next state our main contributions in detail.  We prove new
 \CONGEST upper bounds for approximating minimum weighted vertex cover
 and maximum weighted matching (MWM). We start by describing our
 results for the vertex cover problem. In \cite{opodis20_MVC}, it was
 shown that in bipartite graphs, the unweighted vertex cover problem
 can be $(1+\eps)$-approximated in
 $\poly\log\big(\frac{\log n}{\eps}\big)$ time in the \CONGEST
 model. The following theorem is a generalization of the result of
 \cite{opodis20_MVC} to weighted graphs.

\begin{theorem}\label{thm:approxMWVCbipartite}
  For every $\eps\in(0,1]$, there is a \emph{deterministic} \CONGEST
  algorithm to compute a $(1+\eps)$-approximation for the minimum
  weighted vertex cover problem in bipartite graphs in time
  $\poly\big(\frac{\log n}{\eps}\big)$.
\end{theorem}

The next theorem shows that in graph families that are closed under
taking (induced) subgraphs and in which we can efficiently compute
large (or heavy) independent sets, we can efficiently approximate
minimum (weighted) vertex cover with an approximation ratio that is
better than $2$.

\begin{theorem}\label{thm:approxMWVCgeneral}
  Let $\calG$ be a family of weighted graphs  that is closed under taking induced subgraphs and such that for some
  $\lambda\in (0,1]$ and any $n$-node graph $G=(V,E,w)$ of $\calG$,
  there is a $T_\lambda(n)$-round \CONGEST algorithm to compute an
  independent set $S$ of weight $w(S)\geq \lambda w(V)$. Then, there is
  $T_\lambda(n)+\poly\big(\frac{\log n}{\eps}\big)$-round \CONGEST
  algorithm to compute a $(2-2\lambda+\eps)$-approximate weighted vertex cover for graphs of $\calG$. If the
  independent set algorithm is deterministic, then also the vertex
  cover algorithm is deterministic.
\end{theorem}

Note that the algorithm of \Cref{thm:approxMWVCgeneral} uses the bipartite
vertex cover algorithm of \Cref{thm:approxMWVCbipartite} as a subroutine.
\Cref{thm:approxMWVCgeneral} in particular implies that for graphs for
which we can compute a coloring with a small number of colors, we can
efficiently compute a non-trivial vertex cover approximation.

\begin{corollary}\label{cor:approxMWVCcoloring}
  Let $\calG$ be a family of weighted graphs such that for some
  non-negative integer $C$, for any $n$-node graph $G=(V,E,w)$ of
  $\calG$, there is a $T_C(n)$-round \CONGEST algorithm to compute a
  vertex coloring of $G$ with $C$ colors. Then, there is a
  $T_C(n) +\poly\big(\frac{\log n}{\eps}\big)$-round \CONGEST
  algorithm to compute a $(2-2/C+\eps)$-approximation of the minimum
  weighted vertex cover problem for graphs of $\calG$. If the coloring
  algorithm is deterministic, then also the vertex cover algorithm is
  deterministic.
\end{corollary}

In order to efficiently compute an independent set $S$ of
weight $w(S)\geq w(V)/C$ from a $C$-coloring, we need the graph to be
of small diameter. However, by using standard clustering techniques
(which we anyways need to apply also for our bipartite vertex cover
algorithm), one can reduce the minimum (weighted) vertex cover problem
on general $n$-node graphs to graphs of diameter
$\poly\big(\frac{\log n}{\eps}\big)$. In particular, in
graphs of arboricity $a$, we can (deterministically) compute a
$(2+\eps)a$-coloring in time $O(\log^3 a\cdot \log n)$~\cite{GhaffariKuhn21}. As a consequence, we get a deterministic
$\poly\big(\frac{\log n}{\eps}\big)$-round \CONGEST algorithm for
computing a $(2-1/a + \eps)$-approximation  of minimum weighted vertex
cover in graphs of arboricity $a$.

In addition to our \CONGEST algorithms for approximating minimum
weighted vertex cover, we also provide new \CONGEST algorithms for
approximating maximum weighted matching. The following theorem can be seen as a generalization of Theorem 3.15 in \cite{lotker15} and of Theorem B.12 in \cite{yehuda17} (full version).

\begin{theorem}\label{thm:approxMWM}
  For every $\eps,\delta\in(0,1]$, there is a \emph{randomized}
  \CONGEST algorithm that with probability at least $1-\delta$
  computes a $(1-\eps)$-approximation to the maximum weighted matching
  problem in $2^{O(1/\eps)}\cdot \left(\log (W \Delta) + \log^2\Delta+\log^* n \right)\cdot \log^3(1/\delta)$ rounds. Further, there
  is a \emph{deterministic} \CONGEST algorithm to compute a
  $(1-\eps)$-approximation for the minimum weighted matching problem in time
  $2^{O(1/\eps)}\cdot\poly\log n$.
\end{theorem}

Note that except for the $\log^* n$ term, for constant error
probability $\delta$, the round complexity of our randomized algorithm
is independent of the number of nodes $n$. Moreover, for constant
$\eps$ and $\delta$, in bounded-degree graphs with bounded weights,
the round complexity of the randomized algorithm is only
$O(\log^* n)$. For unweighted matchings, a round complexity that is
completely independent of $n$ was obtained by
\cite{yehuda17}. Interestingly, G\"o\"os and Suomela in
\cite{goeoes14_DISTCOMP} showed that such a result is not possible for
the minimum vertex cover problem, even in the \LOCAL model. They show
that even for bipartite graphs of maximum degree $3$, there exists a
contant $\eps_0>0$ such that any randomized distributed
$(1+\eps_0)$-approximation algorithm for the (unweighted) minimum
vertex cover problem requires $\Omega(\log n)$ rounds. As our last
contribution, we generalize the result of \cite{goeoes14_DISTCOMP} to
computing $(1+\eps)$-approximate solutions for any sufficiently small
$\eps>0$.

\begin{theorem}\label{thm:lowerbound}
  There exists a constant $\eps_0>0$ such that for every
  $\eps\in(0,\eps_0]$, any randomized \LOCAL model algorithm to
  compute a $(1+\eps)$-approximation for the (unweighted) minimum vertex cover
  problem in bipartite graphs of maximum degree $3$ requires
  $\Omega\big(\frac{\log n}{\eps}\big)$ rounds.
\end{theorem}

\Cref{thm:lowerbound} is obtained by a relatively simple reduction to
the $(1+\eps_0)$-approximation lower bound proven in
\cite{goeoes14_DISTCOMP} for bipartite graphs of maximum degree $3$.
We note that if we only require the approximation factor to hold in
expectation, as discussed at the end of \Cref{sec:diameter_reduction},
in the \LOCAL model the theorem is tight even for general graphs and
even for the weighted vertex cover problem.

\medskip

\para{Organization of the paper:} The remainder of the paper is
organized as follows. In \Cref{sec:model}, we define the communication
model and we introduce all the necessary mathematical notations and
definitions. In \Cref{sec:overview}, we give an overview over all our
algorithms and our most important ideas and techniques. In
\Cref{sec:MWVC}, we provide the additional technical details needed
for the weighted vertex cover algorithms, \Cref{sec:MWM} is devoted to
the details of the maximum weighted matching algorithms, and in
\Cref{sec:lower}, we formally prove the lower bound on approximating
vertex cover in bipartite graphs in the \LOCAL model. Finally, in
\Cref{sec:tools}, we describe some basic algorithmic tools that we
need for our algorithms and which already appear in the literature in
a very similar form.


\section{Model and Preliminaries}
\label{sec:model}

\subsection{Mathematical Notation}
\label{sec:notation}
Let $G=(V,E,w)$ be an undirected weighted graph, where $w$ is a
non-negative weight function. We will use node and edge weights in the
paper and depending on the context, we will use $w$ to assigns weights
to nodes and/or edges. Generally for a set $X$ of nodes and/or edges,
we use $w(X)$ to denote the sum of the weights of all nodes/edges in
$X$. For example, if we have node weights, $w(V)$ denotes the sum of
the weights of all the nodes. Throughout the paper, we assume that all
weights are integers that are polynomially bounded in the number of
node of the graph. However, as long as we can communicate a single
weight in a single message, all our algorithms can be adapted to also
work at no significant additional asymptotic cost for more general
weight assignments. We further use the following notation for
graphs. For a node $v\in V$, we use $N(v)\subseteq V$ to denote the set
of neighbors of $v$ and we use $E(v)\subseteq E$ to denote the set of
edges that are incident to $v$.

For a graph $G=(V,E)$, the \emph{bipartite double cover} is defined as the
graph $G_2 := G \times K_2 = (V \times \{0, 1\} , E_2)$, where there
is an edge between two nodes $(u, i)$ and $(v, j)$ in $E_2$ if and
only if $\{u, v\} \in E$ and $i \neq j$. Hence, in $G_2$, every node
$u$ of $G$ is replaced by two nodes $(u, 0)$ and $(u, 1)$ and every
edge $\{u, v\}$ of $G$ is replaced by the two edges $\{(u, 0),(v,
1)\}$ and $\{(u, 1),(v, 0)\}$. Moreover, if $G$ is a weighted graph
with weight function $w$, we assume that the bipartite double cover
$G_2$ is also weighted and that the corresponding nodes and/or edges have the same
weight as in $G$. That is, in case of node weights, for every $u\in
V$, we define $w((u,0))=w((u,1))=w(u)$ and in case of edge weights,
for every $\{u,v\}\in E$, we define $w(\{u,i\},\{v,1-i\})=w(\{u,v\})$
for $i\in\{0,1\}$.

\subsection{Problem Definitions}
\label{sec:defProblems}

In this paper, we consider the \emph{minimum weighted vertex cover
  (MWVC)} and the \emph{maximum weighted matching (MWM)}
problems. Formally, in the MWVC problem, we are given a weighted graph
$G=(V,E,w)$ with positive node weights. A vertex cover of $G$ is a set
$S\subseteq V$ of nodes such that for every edge $\set{u,v}\in E$,
$S\cap\set{u,v}\neq\emptyset$. The goal of the MWVC problem is to find
a vertex cover $S$ of minimum total weight $w(S)$. In the MWM problem,
we are given a weighted graph $G=(V,E,w)$ with positive edge
weights. A matching of $G$ is a set $M\subseteq E$ of edges such that
no two edges in $M$ are adjacent. The goal of the MWM problem is to
find a matching $M$ of maximum total weight $w(M)$. The unweighted
versions of the two problems are closely related to each other as
their natural fractional linear programming (LP) relaxations are duals
of each other. In the paper, we will also use the fractional
relaxation of the MWVC problem and its dual problem. In the fractional
MWVC problem on $G$, every node $u\in V$ is assigned a value
$x_u\in[0,1]$ such that for every edge $\set{u,v}\in E$,
$x_u+x_v\geq 1$ and such that the sum $\sum_{u\in V}w(u)\cdot x_u$ is
minimized. The dual LP of this problem, which for a given weight
function $w$, we in the following call the \emph{fractional
  $w$-matching} problem, is defined as follows. Every edge $e\in E$ is
assigned a (fractional) value $y_e\geq 0$ such that for every node
$u\in V$, we have $\sum_{e:u\cap e\neq \emptyset}y_e \leq w(u)$ and
such that the sum $\sum_{e\in E} y_e$ is maximized. We use the vector
$\vec{y}$ to refer to a fractional solution that assigns a fractional
value $y_e$ to every edge. Further for convenience, for a set of
edges $F$, we also use the short notation $y(F):=\sum_{e\in F}
y_e$. LP duality directly implies that the value of any fractional
$w$-matching cannot be larger than the weight of any vertex
cover:

\begin{lemma}\label{lemma:duality}
  Let $G=(V,E,w)$ be a node-weighted graph and let $\vec{y}$
  be a fractional $w$-matching of $G$. It then holds that $y(E)\leq
  w(S)$ for every vertex cover $S$ of $G$.
\end{lemma}
\begin{proof}
  We have
  \[
    w(S) = \sum_{v\in S} w(v) \geq \sum_{v\in S}\sum_{e:v\in e} y_e
    \geq \sum_{e\in E} y_e = y(E).
  \]
  The first inequality holds because the values $y_e$ form a valid
  fractional $w$-matching and the second inequality holds because $S$
  is a vertex cover.
\end{proof}

The \emph{approximation ratio} of an
approximation algorithm for the MWVC or MWM problem is defined as the
worst-case ratio between the total weight of a vertex cover or
matching computed by the algorithm over the total weight of an optimal
vertex cover or matching. That is, we define the approximation ratio
such that it is $\geq 1$ for minimization and $\leq 1$ for
maximization problems.

\subsection{Low-Diameter Clustering}
\label{sec:defClustering}

Many of our algorithms have some components that require global
communication in the network. In order the achieve a polylogarithmic
round complexity, we therefore need a graph with polylogarithmic
diameter. We achieve this by applying standard clustering
techniques. Formally, we use the clusterings as in
\cite{opodis20_MVC} described in the following. Let $G=(V,E,w)$ be a weighted graph with
non-negative node and edge weights. A \emph{clustering} of $G$ is a
collection $\set{S_1,\dots,S_k}$ of disjoint node sets
$S_i\subseteq V$. For $\lambda\in[0,1]$, a clustering
$\set{S_1,\dots,S_k}$ is called \emph{$\lambda$-dense}
 if the total
weight of all nodes and edges in the induced subgraphs $G[S_i]$ for
$i\in\set{1,\dots,k}$ is at least $\lambda(w(V)+w(E))$. Further, for
an integer $h\geq 1$, a clustering $\set{S_1,\dots,S_k}$ is called
\emph{$h$-hop separated} if for any two clusters $S_i$ and $S_j$
($i\neq j$) and any pair of nodes $(u,v)\in S_i\times S_j$, we have
$d_G(u,v)\geq h$, where $d_G(u,v)$ denotes the hop-distance between
$u$ and $v$. Further for two integers $c,d\geq 1$, a clustering
$\set{S_1,\dots,S_k}$ is defined to be \emph{(c,d)-routable} if we are
given a collection of $T_1,\dots,T_k$ trees in $G$ such that for each
$i\in \set{1,\dots,k}$, the nodes $S_i$ are contained in $T_i$, every
tree $T_i$ has diameter at most $d$, and every edge $e\in E$ of $G$ is
contained in at most $c$ of the trees $T_1,\dots,T_k$. Note that this
implies that each cluster of a $(c,d)$-routable clustering has weak
diameter at most $d$ and if the nodes of $T_i$ are all contained in
$S_i$, it implies that the strong diameter of cluster $S_i$ is at most
$d$.

\subsection{Communication Model}
Throughout the paper, we assume a standard synchronous message passing
model on graphs. That is, the network is modeled as an undirected $n$-node
graph $G=(V,E)$. Each node is equipped with a unique $O(\log n)$-bit
identifier. The nodes $V$ communicate in synchronous rounds over
the edges $E$ such that in each round, every node can send an
arbitrary mesage to each of its neighbors. Internal computations at
the nodes are free. Initially, the nodes do not know anything about
the topology of the network. When computing a vertex cover or a
matching, at the end of the algorithm, every node must output if it is
in the vertex cover or which of its edges belong to the matching. The
time or round complexity of an algorithm is defined as the number of
rounds that are needed until all nodes terminate. If the size of the
messages is not restricted, this model is known as the \LOCAL
model~\cite{peleg00}. In the more restrictive \CONGEST model, all messages must
consist of at most $O(\log n)$ bits~\cite{peleg00}. In several of our
algorithms, we will first compute a clustering as defined above in
\Cref{sec:defClustering} and we afterwards run \CONGEST algorithms on
the clusters. If we are given a $(c,d)$-routable 
clustering, we are only guaranteed that the diameter of each cluster
$S_i$ is small if we add the nodes and edges of the tree $T_i$ to the
cluster. For running our algorithms on individual clusters, we
therefore need an extension of the classic \CONGEST model, which has
been introduced as the \SUPPORTED model in
\cite{FoersterKR019,SchmidS13}. In the \SUPPORTED model, we are given
two graphs, a communication graph $H=(V_H,E_H)$ and a logical graph
$G=(V,E)$, which is a subgraph of $H$. When solving a graph problem
such as MWVC or MWM in the \SUPPORTED model, we need to solve the
graph problem on the logical graph $G$, we can however use \CONGEST
algorithms on the underlying communication graph $H$ to do so. Note
that if we are given a $(c,d)$-routable clustering, we can define
$G_i:=G[S_i]$ and $H_i$ as the union of the graph $G_i$ and the tree
$T_i$ for each cluster and we can then in parallel run $1$ round of a
\SUPPORTED algorithm on each cluster in $c$ \CONGEST rounds on $G$.


\section{Technical Overview}
\label{sec:overview}

In this section, we provide an overview of the core ideas and
techniques for all our results. While we try to provide intuitive
arguments for everything, most of the formal proofs appear in
\Cref{sec:MWVC,sec:MWM,sec:lower}. We start by describing how we can
use clusterings to reduce the problem of approximating MWVC or MWM on
general graphs to the case of approximating the same problems on
graphs of small diameter.

\subsection{Reducing to Small Diameter}
\label{sec:diameter_reduction}

The high-level idea that we use to reduce the diameter is a classic
one. We find a disjoint and sufficiently separated collection of
low-diameter clusters such that only a small fraction of the graph is
outside of the clusters~(see, e.g.,
\cite{Awerbuch-Peleg1990,linial93,MPX13,peleg00,polylogdecomp} for
constructions of such clusterings). We then compute a good
approximation for a given problem inside each cluster and we use a
coarse approximation to extend the solution to the parts of the graph
outside of the clusters. If we want this to work in general graphs, we
have to adapt the standard clustering constructions such that the part
of the graph that is outside of the clusters contains only a
\emph{small fraction of a solution to the actual problem} that we want
to approximate, rather than simply a small fraction of the number of
nodes and/or edges of the graph. A generic way to achieve this in the
\LOCAL model has been described in \cite{ghaffari2017complexity} and a
method that can also be used in the \CONGEST model has recently been
described in \cite{opodis20_MVC} for the unweighted minimum vertex cover problem.
The following theorem shows how to extend the approach of
\cite{opodis20_MVC} to work also for weighted vertex cover and
matching.  The theorem shows that at the cost of a $(1+\eps)$-factor,
the problems of approximating MWVC and MWM can efficiently be reduced
to approximating the problems in the \SUPPORTED model
with a small-diameter communication graph.

\begin{theorem}[Diameter Reduction]\label{thm:diameterreduction}
  Let $T_{\mathsf{SC}}^{\alpha}(n,D)$ be the time required for computing an
  $\alpha$-approximation for the MWVC or the MWM problem in the
  \SUPPORTED model with a communication graph of diameter $D$. Then,
  for every $\eps\in(0,1]$, there is a 
  $\poly\big(\frac{\log n}{\eps}\big) + O\big(\log n \cdot
  T_{\mathsf{SC}}^{\alpha}\big(n, O\big(\frac{\log^3 n}{\eps}\big)\big)\big)$-round
  \CONGEST algorithm to compute a $(1-\eps)\alpha$-approximation of
  MWM or an $(1+\eps)\alpha$-approximation of MWVC in the \CONGEST
  model. If the given \SUPPORTED model algorithm is deterministic,
  then the resulting \CONGEST model algorithm is also
  deterministic. Also, if we want to solve MWVC
  or MWM in the \CONGEST model on a bipartite graph, then it is
  sufficient to have a \SUPPORTED model algorithm that works for a
  bipartite communication (and thus also logical) graph.
\end{theorem}
\begin{proof}
  For both problems, we first compute node or edge weights that we
  will use to compute a good clustering of the graph. Assume that we
  have a graph $G=(V,E,w)$, where for the MWVC problem, $w$ is a
  function that assigns positive integer weights to the nodes and for the MWM
  problem, $w$ is a function that assigns positive integer weight to
  the edges.

  For the MWVC problem, we first compute a fractional $w$-matching
  $\vec{y}$ as follows.\footnote{A similar parallel approximation
    algorithm for MWVC to the best of our knowledge first appeared in
    \cite{KhullerVY94}.} We initialize the value of each edge
  $e=\set{u,v}$ to $y_e:=\min \{w(u),w(v)\}/\Delta$. Note that this is a
  feasible fractional $w$-matching. Now, we improve the fractional
  $w$-matching as follows in phases. We call a node $v$ half-tight if
  $\sum_{e\in E(v)} y_e > w(v)/2$. In each phase, we let $E_f$ be the
  set of edges $\set{u,v}$ for which both $u$ and $v$ are not
  half-tight and we double the $y_e$-value of all edges in $E_f$. Like
  this, $\vec{y}$ remains a feasible fractional $w$-matching and after
  $O(\log \Delta)$ phases, we obtain a fractional $w$-matching
  $\vec{y}$ such that at least one node of every edge is
  half-tight. Hence, the set $S$ of half-tight nodes is a vertex cover
  of weight
  \begin{equation}\label{eq:wmatchingapprox}
    w(S) \leq \sum_{v\in S} w(v) \leq \sum_{v\in S}2\sum_{e\in
      E(v)}y_e \leq 4y(E) \leq 4 w(S^*),
  \end{equation}
  where $S^*$ is an optimal weighted vertex cover of $G$. The last
  inequality follows from \Cref{lemma:duality}. Further, clearly, one
  phase requires $O(1)$ rounds and we therefore need $O(\log n)$
  rounds for computing the fractional $w$-matching $\vec{y}$.

  For the MWM problem, we proceed as follows. The LP dual of the
  fractional relaxation of the MWM problem asks for assigning a value
  $x_v\geq 0$ to each node such that $x_u+x_v\geq w(e)$ for every edge
  $e=\set{u,v}$ and such that $\sum_{v\in V} x_v$ is minimized. This
  is a fractional covering problem and we can for example compute a
  $4$-approximation for it by using the \CONGEST algorithm of
  \cite{nearsighted}. For the given covering problem, the round
  complexity of the algorithm of \cite{nearsighted} is $O(\log(\Delta
  W))$ if $W$ is the largest edge weight. Since we assumed that $W$ is
  at most polynomial in $n$, the round complexity for computing the
  assignment of $x_v$-values if $O(\log n)$.
  
  In both cases, we now compute a $(1-\eps/4)$-dense, $3$-hop
  separated clustering of $G$ by using \Cref{thm:det_clustering}. In
  the case of MWVC, we use the fractional $w$-matching $\vec{y}$ as
  edge weights for the clustering and we set all the node weights to $0$. In the case of MWM, we use the
  assignment of $x_v$-values as node weights and we set all the edge weights to $0$. By
  \Cref{thm:det_clustering}, the time for computing such a clustering
  is $\poly\big(\frac{\log n}{\eps}\big)$. Let $S_1,\dots,S_t$ be
  the collection of clusters returned by the clustering algorithm. In
  both cases, we define extended clusters $S_1',\dots,S_t'$, where
  cluster $S_i'$ consists of all nodes in $S_i$ and of all neighbors
  of nodes in $S_i$. Note that because of the $3$-hop separation of
  clusters, the clusters $S_i'$ are still vertex-disjoint. By
  \Cref{thm:det_clustering}, the clustering is
  $\big(O(\log n), O\big(\frac{\log^3
    n}{\eps}\big)\big)$-routable. On the induced subgraphs $G[S_i']$
  of the extended clusters, we can therefore run \SUPPORTED algorithms
  with a communication graph of diameter
  $O\big(\frac{\log^3 n}{\eps}\big)$. If we run such algorithm on all
  clusters in parallel, we have a slowdown of $O(\log n)$. By the
  assumption of the lemma, we can therefore in parallel for all graph
  $G[S_i']$ compute an $\alpha$-approximate solution for our given
  problem (MWVC or MWM) in time
  $O\big(\log n \cdot T_{\mathsf{SC}}^{\alpha}\big( O\big(\frac{\log^3
    n}{\eps}\big)\big)$.

  For weighted matchings, we are now done. By LP duality, the maximum
  weight matching for a subgraph induced by a subset
  $F\subseteq E$ of the edges of $G$ is upper bounded by the sum of
  the $x_v$ values assigned to all the nodes incident to edges in
  $F$. Hence the maximum weight matching of the edges that are outside
  the graphs $G[S_i']$ is upper bounded by the sum of the $x_v$ values
  of nodes outside the original clusters $S_i$ and because the
  clustering is $(1-\eps/4)$-dense, we know that the sum of those
  $x_v$-values is at most $\eps/4\cdot\sum_{v\in V} x_v$. Because the
  assignment $x_v$ constitutes a $4$-approximation of the dual of the
  fractional weighted matching problem, we also know that the weight
  of an optimal matching is at least $1/4\cdot\sum_{v\in V}
  x_v$. The claim of the theorem thus follows for the MWM problem.

  For the MWVC problem, we know that the collection of the
  $\alpha$-approximate weighted vertex covers of the graphs $G[S_i']$
  is clearly upper bounded by $\alpha w(S^*)$, where $w(S^*)$ is the
  weight of an optimal weighted vertex cover of the whole graph
  $G$. However, the collection of the vertex covers of the clusters is
  not a valid vertex cover of $G$. Some of the edges outside clusters
  $S_i'$ might not be covered. The edges outside clusters $S_i'$ can
  however be covered by the set of half-tight nodes (w.r.t.\ the
  fractional $w$-matching $\vec{y}$) outside the original clusters
  $S_i$. The total weight of those nodes is at most $4$ times the
  fractional values $y_e$ of all the uncovered edges and thus at most
  $\eps \cdot y(E)\leq \eps w(S^*)$ (because of
  \eqref{eq:wmatchingapprox} and because the clustering is
  $(1-\eps/4)$-dense).
\end{proof}

\para{Remark.} We note that by using a variant of the randomized clustering
algorithm of Miller, Peng, and Xu~\cite{MPX13}, one can compute a
$\big(1,O\big(\frac{\log n}{\eps}\big)\big)$-routable, $3$-hop separated clustering with
expected density $1-\eps$ in time $O\big(\frac{\log n}{\eps}\big)$ (also cf.\ \Cref{sec:clustering} and \cite{opodis20_MVC}). In
combination with the argument in the above proof, this in particular
implies that it in $O\big(\frac{\log n}{\eps}\big)$ rounds in the \LOCAL model, it is
possible to compute weighted matchings and weighted vertex covers with
expected approximation ratios $1-\eps$ and $1+\eps$, respectively.

\subsection{Basic Bipartite Weighted Vertex Cover Algorithm}
\label{sec:basicbipartiteMWVC}

It is well-known that in bipartite graphs, the size of a maximum
matching is equal to the size of a minimum vertex cover (this is known
as K\H{o}nig's theorem~\cite{koenig_diestel,koenig31}). The theorem
was also independently discovered in \cite{egervary31} by
Egerv\'{a}ry, who also more generally proved that on node-weighted
bipartite graphs, the total value of an optimal (fractional)
$w$-matching is equal to the weight of a minimum weighted vertex
cover, where a fractional $w$-matching of a node-weighted graph
$G=(V,E,w)$ is an assignment of fractional values $y_e\geq 0$ to all
edges such that the edges of each node $v$ sum up to at most
$w(v)$. In both cases, the theorem can be proven in a constructive
way. Given a maximum matching or more generally a maximum fractional
$w$-matching, there is a simple (and efficient) algorithm to compute a
vertex cover of the same size or weight.

Moreover as shown in 
\cite{FMS15}, if we are given a good approximate matching or
$w$-matching with some additional properties, the constructive proof
of \cite{egervary31,koenig31} can be adapted to obtain a good
approximate (weighted) vertex cover. For the unweighted case,
this method is at the core of the \CONGEST model bipartite vertex cover
algorithms of \cite{opodis20_MVC}. We next describe how
to use this technique to approximate MWVC in the \CONGEST model.

Let $G=(V,E,w)$ be a node-weighted graph, where $w$ is a non-negative
node weight function. For a node $v\in V$ and a given fractional
$w$-matching $y_e$ for $e\in E$, we define the \emph{slack} of $v$ as
$s(v):=w(v)-\sum_{e:v\in e}y_e$. Given a fractional $w$-matching $y_e$
for $e\in E$, an \emph{augmenting path} is an odd length path
$P=(v_0,\dots,v_{2k+1})$ where the end nodes $v_0$ and $v_{2k+1}$ have
positive slack $s(v_0),s(v_{2k+1})>0$, and for $i \in \{1,2,...,k\}$,
each even edge $e=(v_{2i-1},v_{2i})$ has a positive fractional value
$y_e>0$. Assume that we are given a bipartite graph $G=(A\cup B,E,w)$
and that we are given a fractional $w$-matching $\vec{y}$ of $G$ such
that there are no augmenting paths of length at most $2k-1$ for some
integer $k\geq 1$. We can then apply the following algorithm to
compute a vertex cover $S$ of $G$. The algorithm computes disjoint
sets $A_0,A_1,\dots,A_k\subseteq A$ and disjoint sets
$B_1,\dots,B_k\subseteq B$. In the following for a set of nodes $X$,
we use $N_{y>0}(X)$ to denote the set of nodes that are connected to a
node in $X$ through an edge $e$ with $y_e>0$.

\bigskip

\begin{center}
  \begin{minipage}{1.0\linewidth}
    \vspace{-8pt}
    
    \begin{mdframed}[hidealllines=false, backgroundcolor=gray!10]
      \textbf{Basic Approximate Weighted Vertex Cover
        Algorithm}
      \vspace*{-8pt}
      \begin{enumerate}
      \item Define $A_0 :=\set{v\in A\,:\,s(v)>0}$ as the nodes in $A$
        with positive slack and $B_0:=\emptyset$.
      \item For every $i\in \set{1,\dots,k}$, define $B_i:=\set{v\in 
        	B\setminus \bigcup_{j=0}^{i-1} B_j\,:\,v\in N(A_{i-1})}$.
        
      \item For every $i\in \set{1,\dots,k}$, define $A_i:=\set{v\in 
      	A\setminus \bigcup_{j=0}^{i-1} A_j\,:\,v\in N_{y>0}(B_i)}$.

      \item Define $i^* := \arg\min_{i\in
          \{1,\dots,k\}} w(B_i)$.
      \item Output $S:=\bigcup_{i=1}^{i^*} B_i \cup
        \big(A \setminus \bigcup_{i=0}^{i^*-1} A_i\big)$.
      \end{enumerate}
    \end{mdframed}
  \end{minipage}
  \vspace{-8pt}
\end{center}

\smallskip
That is, the sets $A_0,B_1,A_1,B_2,A_2\dots$ are the levels of a BFS
traversal of the graph starting at the nodes in $A_0$ and where steps
from $B_i$ to $A_i$ have to be over an edge $e$ with positive
fractional value $y_e>0$.

\begin{lemma}\label{lemma:basicMWVCalg}
  Given a weighted bipartite graph $G=(A\cup B, E, w)$,
  an integer $k\geq 1$,
  and a fractional $w$-matching of $G$ with no augmenting paths of
  length at most $2k-1$, the above algorithm computes a
  $(1+1/k)$-approximate weighted vertex cover of $G$. Further, the
  above algorithm can be deterministically implemented in $O(D+k)$
  rounds in the \SUPPORTED model if the communication graph is also
  bipartite and has diameter at most $D$.
\end{lemma}
\begin{proof}
  We first prove that $S$ is a valid vertex cover. For this, we need
  to show that there is no edge between $A\setminus S$ and
  $B \setminus S$. We have $A\setminus S = \bigcup_{j=0}^{i^*-1} A_j$
  and $B\setminus S=\bigcup_{j=i^*+1}^{k} B_j\cup \overline{B}$, where
  $\overline{B}=B\setminus \bigcup_{i=1}^k B_i$. We therefore need to
  show that there cannot be an edge between a set $A_j$ for $j<i^*$
  and $\sum_{j=i^*+1}^k B_j\cup \overline{B}$. However, by
  construction, all neighbors of nodes in $A_j$ are in
  $B_1,\dots,B_{i^*}$ and we can thus conclude that $S$ is a vertex
  cover.

  We next show that $S$ is a $(1+1/k)$-approximate weighted vertex
  cover of $G$. First observe that for all $i \in \{1,2..,k\}$, all
  nodes $v$ in $B_i$ are saturated nodes with zero slack, i.e., $
  w(v)= \sum_{e\in E(v)}y_e$. From the construction of the sets
  $A_0,B_1,A_1,\dots$, we otherwise get an augmenting path of length
  at most $2k-1$. Moreover since the sets of edges incident to the sets
  $B_1,\dots,B_k$ are disjoint, the sum of the fractional values of
  all those edges is at most $y(E)=\sum_{e\in E} y_e$. By the choice
  of $i^*$, we therefore know $w(B_{i^*})=\sum_{v\in
    B_{i^*}}\sum_{e\in E(v)} y_e \leq y(E)/k$. Moreover, from the BFS
  construction of the sets $A_i$ and $B_i$ it follows that the
  only edges $e$ with $y_e>0$ for which both nodes are in $S$ are
  edges that are incident to nodes in $B_i^*$. We can therefore
  conclude that $w(S)\leq y(E) + w(B_{i^*})\leq (1+1/k)\cdot
  y(E)$. The approximation ratio now follows from LP duality (i.e.,
  from \Cref{lemma:duality}).

  Finally, we discuss how the above algorithm can be efficiently
  implemented in time $O(D+k)$ rounds in the \SUPPORTED model, where
  $D$ is the diameter of the communication graph $H$. In $O(D)$ rounds,
  one can compute a BFS spanning tree of $H$ and use it to compute the
  bipartition of the nodes of $H$ and thus of $G$ into sets $A$ and
  $B$. Then in $O(k)$ rounds, the algorithm constructs the sets $A_i$ and $B_i$ for $i \in \{1,2..,k\}$ by running the first $2k$
  iterations of parallel BFS on $G$ starting from set $A_0$ (where edges
  from $B_i$ to $A_i$ need to have positive fractional values).
  Finally in another $O(D+k)$ rounds, we use the precomputed BFS spanning tree
  on $H$ and a standard pipelining scheme for the root to compute 
  the weights of all the sets $B_i$, determine the index $i^*$ of the
smallest weight amongst them, and broadcast it to all nodes in $G$. 
\end{proof}

We remark that although the above algorithm only requires a BFS traversal from all nodes in $A_0$ for $k$ levels, the algorithm still requires time $D$ for two reasons. First, the algorithm needs to know the bipartition $A\cup B$ of $G$ and computing the bipartition requires $\Omega(D)$ time. Second, even if the bipartition is given initially, the algorithm still needs $\Omega(D)$ time to determine the optimal level $i^*$.

\subsection{Getting Rid of Short Augmenting Paths}
\label{sec:noshortpaths}
The basic MWVC algorithm described in \Cref{sec:basicbipartiteMWVC}
basically converts a good approximation of fractional $w$-matching
into a good MWVC approximation. However the algorithm needs a
fractional $w$-matching with the additional property that there are no
short augmenting paths. For the unweighted setting, there exists a
randomized \CONGEST algorithm to compute an integral matching with no
short augmenting paths~\cite{lotker15}. It is however not clear if the
algorithm of \cite{lotker15} can be generalized to the fractional
$w$-matching problem. Further, even in the unweighted case, we do not
have a deterministic \CONGEST algorithm to compute such a matching. As
in the deterministic, unweighted MVC algorithm of \cite{opodis20_MVC},
we therefore use a different approach. In the unweighted setting, we
first compute a $(1-\delta)$-approximate matching that can potentially
have short augmenting paths. We then get rid of those short augmenting
paths by removing at least one unmatched node or both nodes of a
matching edge from the graph. The removed nodes are at the end added
to the vertex cover to make sure that all edges are covered. The
selection of a smallest possible number of unmatched nodes and
matching edges that hit all short augmenting paths can be phrased as a
minimum set cover problem, which we can approximate efficiently in the
\CONGEST model. In the weighted case, we use a generalization of this
approach. Because of the weights, the process and its analysis however
becomes more subtle and we have to be more careful.

 Assume that we want to compute a $(1+O(\eps))$-approximate
weighted vertex cover for a node-weighted bipartite graph
$G=(A\cup B, E, w)$. In a first step, we compute a
$(1-\delta)$-approximate fractional $w$-matching
$\vec{y}:=\set{y_e : e\in E}$ of $G$ for some parameter
$\delta\ll\eps$. We can do this efficiently by using
\Cref{thm:fractional_wmatching}. The fractional $w$-matching $\vec{y}$
however might have short augmenting paths. In a second step, we then
convert our graph $G$ and the fractional $w$-matching $\vec{y}$ such
that we obtain an instance with no short augmenting paths and that we
can thus apply \Cref{lemma:basicMWVCalg}. More concretely, we
decrease some of the weights $w(v)$ and some of the fractional values
$y_e$ such that for the resulting weights $w'(v)$ and the resulting
$w'$-matching $\vec{y}'$, the graph $G$ has no short augmenting paths
and such that a $(1+\eps)$-approximate weighted vertex cover of $G$
with the weights $w'(v)$ is a $(1+O(\eps))$-approximate weighted
vertex cover of $G$ for the original weights. We next describe the
main ideas of this transformation.

Formally, the conversion can be defined by a set
$X\subseteq \set{v\in A\cup B : s(v)>0}$ of nodes with positive slack
and a set $F\subseteq \set{e\in E : y_e>0}$ of edges with positive
fractional value. The new fractional values $y_e'$ and the new weights
$w'(v)$ are defined as follows:
\begin{equation}
  \label{eq:conversion}
  y_e':=
  \begin{cases}
    0 & \text{if } e\in F\\
    y_e & \text{if } e\not\in F
  \end{cases}\ ,
  \quad\quad\quad
  w'(v) :=
  \begin{cases}
    w(v) - s(v) - y(E(v)\cap F) & \text{if }  v\in X  \\
    w(v) - y(E(v)\cap F) & \text{if } v\notin X
  \end{cases}
\end{equation}

\begin{lemma}\label{lemma:conversion}
  Any augmenting path of $G$ w.r.t.\ the weight function $w'$ and the
  fractional $w'$-matching $\vec{y}'$ is also an augmenting path
  w.r.t.\ the original weight function $w$ and the original fractional
  $w$-matching $\vec{y}$.
\end{lemma}
\begin{proof}
  First note that whenever we decrease a value $y_e$ to
  $y_e'<y_e$ for some edge $e=\set{u,v}$, we also decrease the weights
  of $u$ and $v$ by the
  same amount. Therefore, the slack of a node $v$ w.r.t.\ $w'$ and $\vec{y}'$
  cannot be larger than the slack of $v$ w.r.t.\ $w$ and
  $\vec{y}$. Therefore any odd-length path $P$ in $G$ that starts and ends at a node
  with positive slack w.r.t.\ $w'$ and $\vec{y}'$ also starts and
  ends at a node with positive slack w.r.t.\ $w$ and
  $\vec{y}$. Further, if every even edge $e$ of such a path $P$ has a
  positive fractional value $y_e'>0$, then it also holds that $y_e>0$.
\end{proof}

Note that the definition of $w'$ and $\vec{y}'$ in
\eqref{eq:conversion} guarantees that all nodes $v\in X$ have slack
$0$ w.r.t.\ the new weights $w'$ and the new fractional values
$\vec{y}'$. Consider some augmenting path $P=(v_0,\dots,v_{2\ell+1})$
of $G$ w.r.t.\ $w$ and $\vec{y}$. If we have $v_0\in X$ or
$v_{2\ell+1}\in X$ or if we have $e\in F$ for one of the even edges
$\set{v_{2i-1},v_{2i}}$ of $P$, then $P$ is not an augmenting path of
$G$ w.r.t.\ $w'$ and $\vec{y}'$. In order to get rid of all short
augmenting paths, we therefore need to choose one of the end nodes or
one of the even edges of each such path and add them to $X$ or $F$. We
can then use \Cref{lemma:basicMWVCalg} to efficiently compute a good
vertex cover approximation for $G$ w.r.t.\ the new weights
$w'$. The quality of such a vertex cover w.r.t.\ the original weights
$w$ can be bounded as follows.

\begin{lemma}\label{lemma:approxpreserved}
  Let $S^*$ be an optimal weighted vertex cover of  $G=(V,E)$ w.r.t.\
  the weights $w$ and assume that for some $\alpha\geq 1$, $S$ is an
  $\alpha$-approximate weighted vertex cover of $G$ w.r.t.\ the weights
   $w'$. It then holds that
  \[
    w(S) \leq \alpha \cdot w(S^*) + s(X) + y(F),\qquad
    \text{where } s(X):=\sum_{v\in X} s(v).
  \]
\end{lemma}
\begin{proof}
  Let $S'$ be an optimal weighted vertex cover of $G$ w.r.t.\ the
  weights $w'$. Because any vertex cover must contain at least one
  node of every edge in $F$, we have $w'(S')\leq w(S^*)-y(F)$. We
  therefore have
  \[
    w(S) \leq w'(S) + s(X) + 2y(F) \leq \alpha w'(S') + s(X) +
    2y(F)
    \leq \alpha w(S^*) + s(X) + y(F).\qedhere
  \]
\end{proof}

In order to optimize the approximation, we thus need to determine the
sets $X$ and $F$ such that $s(X)+y(F)$ is as small as possible and
such that we `cover' all short augmenting paths. The problem of
finding the best possible sets $X$ and $F$ can naturally be phrased as a
weighted set cover problem. One can further show that if the parameter
$\delta$ that determines the quality of the fractional $w$-matching
$\vec{y}$ is chosen sufficiently small (but still as
$\delta=\poly(\eps/\log n)$), even a logarithmic approximation to this
weighted set cover instance guarantees that
$s(X)+y(F)=O(\eps\cdot w(S^*))$. Further, by sequentially going over the possible short
augmenting path lengths and adapting existing algorithms of
\cite{lotker15} and \cite{opodis20_MVC}, a variant of the greedy
algorithm for this weighted set cover instance can be implemented
efficiently in the \CONGEST model on $G$. Given an efficient algorithm
to find appropriate sets $X$ and $F$, the claim of
\Cref{thm:approxMWVCbipartite} then follows almost immediately by
combining with \Cref{thm:diameterreduction} and
\Cref{lemma:basicMWVCalg}. The details appear in \Cref{sec:MWVC}.

\subsection{Generalization to Non-Bipartite Graphs}
\label{sec:overviewMWVCgeneral}

For approximating the MWVC problem in general graphs, we employ a
standard approach that is for example described in\cite{hochbaum97}. We describe the details for completeness. The minimum fractional (weighted) vertex
cover problem is the natural LP relaxation of the minimum (weighted)
vertex cover problem. That is, a fractional vertex cover of a graph
$G=(V,E)$ is an assignment of values $x_v\in [0,1]$ to all nodes such
that for every edge $\set{u,v}$, $x_u+x_v\geq 1$. While the
integrality gap of the (weighted and unweighted) vertex cover can be
arbitrarily close to $2$, it is well-known that there are optimal
fractional solutions that are \emph{half-integral}
That is, there are
optimal fractional solutions such that for all nodes $v$,
$x_v\in\set{0,1/2,1}$. Given such a fractional solution, let $S_x$ be
the set of nodes $v$ with $x_v=x$ for $x\in \set{0,1/2,1}$. Let $I_{1/2}$
be an independent set of the induced subgraph $G[S_{1/2}]$ of the
half-integral nodes. It is not hard to see that
$S:=S_1 \cup S_{1/2}\setminus I_{1/2}$ is a vertex cover of $G$ and if
$w(I_{1/2})\geq\lambda \cdot w(S_{1/2})$, then the set $S$ is a
$(2-2\lambda)$-approximate solution for the MWVC problem on $G$ with
weights $w$. If the half-integral fractional solution is only an
$\alpha$-approximate fractional weighted vertex cover, the resulting
approximation is $(2-2\lambda)\cdot \alpha$.

The core to proving \Cref{thm:approxMWVCgeneral} is therefore to first
compute a $(1+\eps)$-approximate half-integral fractional solution for
a given weighted graph $G=(V,E,w)$. The following lemma uses a
standard approach to achieve this by first computing an approximate
solution to the (integral) MWVC problem for the bipartite double cover
$G_2$ of $G$ (see definition in \Cref{sec:model}).

\begin{lemma}\label{lemma:doublecoverVC}
  Let $G=(V,E,w)$ be a weighted graph and let $G_2=(V_2,E_2)$ be the bipartite
  double cover of $G$, where each node $(v,i)\in V_2$ for $v\in V$
  gets assigned weight $w(v)$. Let $S$ be a vertex cover of $G_2$ and
  define $x_v:=|\set{(v,0),(v,1)}\cap S|/2$ for every $v\in V$. If $S$
  is an $\alpha$-approximate weighted vertex cover of $G_2$ for some
  $\alpha\geq 1$, then $\vec{x}=\set{x_v : v\in V}$ is a half-integral
  $\alpha$-approximate fractional weighted vertex cover of $G$.
\end{lemma}
\begin{proof}
  We first show that the weight $\sum_{v\in V} w(v)\cdot z_v$ of an
  optimal fractional weighted vertex cover $\vec{z}$ of $G=(V,E,w)$ is
  exactly half the weight of an optimal fractional weighted vertex
  cover of $G_2$ with weights assigned as defined by the claim of the
  lemma. To see this, let $\vec{z}$ be a fractional vertex cover of
  $G$. We can then get a valid fractional vertex cover of $G_2$ by
  setting $z_{(v,0)}=z_{(v,1)}=z_v$ for every node $v\in V$ of $G$ and
  the corresponding nodes $(v,0)$ and $(v,1)$ in $G_2$. In the other
  direction, for a fractional vertex cover $\vec{z}$ in $G_2$, we
  obtain a fractional vertex cover of $G$ by setting
  $z_v:=(z_{(u,0)}+z_{(u,1)})/2$. Note that because $G_2$ is a bipartite
  graph, an optimal (integral) weighted vertex cover of $G_2$ is also
  an optimal fractional weighted vertex cover of $G_2$ and therefore
  the vertex cover $S$ is an $\alpha$-approximate fractional weighted
  vertex cover of $G_2$. The fractional vertex cover $\vec{x}$ of $G$ as
  given by the lemma statement is of half the weight of $S$ in $G_2$
  and it therefore is an $\alpha$-approximate fractional weighted
  vertex cover of $G$. Clearly, $\vec{x}$ is half-integral.
\end{proof}

The proofs of \Cref{thm:approxMWVCgeneral,cor:approxMWVCcoloring}
appear in \Cref{sec:generaldetails}.

\subsection{Weighted Matching Approximation}
\label{sec:matching}
We provide a randomized and a deterministic \CONGEST
algorithm to approximate the MWM problem. Both algorithms are based on
the following key idea that was developed for the unweighted maximum
matching problem by Lotker, Patt-Shamir, and Pettie~\cite{lotker15}
and that was extended by Bar-Yehuda et al.~\cite{yehuda17}. We
iteratively adapt an initial matching $M_0$ of a given weighted graph
$G=(V,E,w)$ as follows. We repeatedly sample bipartite subgraphs of
$G$ and we then find a good matching on this bipartition to improve
the matching of $G$. In \cite{lotker15,yehuda17}, the matching is
improved by finding short augmenting paths in the sampled bipartite
subgraph and augmenting the existing matching along those
paths. While, as discussed below, also for maximum weighted matching,
it is in principle possible to find augmenting paths and cycles, in
the weighted case, we do not know how to do this efficiently in the
\CONGEST model. Instead, we use the bipartite MWM approximation
algorithm of \cite{ahmadi18} to find a good matching in each sampled
bipartite graph. Unlike when using augmenting paths, this approach can
potentially also lead to a worse matching (if the existing matching is
already a very good matching of the sampled bipartite graph). We will
however see that when computing a sufficiently good approximation in
each sampled bipartition, we can use the approach to improve the
matching of $G$ sufficiently often.

Before explaining our algorithms in more detail, we need to introduce
the notion of augmenting paths and cycles for weighted matchings. Given a
matching $M$, a path or cycle in which the edges alternate between
edges $\in M$ and edges $\not \in M$ is called an \emph{alternating
	path or cycle} w.r.t.\ matching $M$. An alternating path or cycle is
called an \emph{augmenting path or cycle} w.r.t.\ $M$ if swapping the
matching edges with the non-matching edges increases the weight of
the matching. This increase is termed the \textit{gain} of an
augmenting path/cycle w.r.t.\ $M$ (we will omit the qualification
'w.r.t.' if it is clear from the context). By definition, the gain of
an augmenting path is positive. Note that alternating cycles (and
therefore also augmenting cycles) are always of even length.

Let $M^*$ be a maximum weighted matching in $G$. Consider the
symmetric difference $F=M^*\triangle M$ of $M^*$
and some arbitrary matching $M$. The set $F$ consists of (vertex-disjoint)
alternating paths and cycles where the edges alternate between $M$ and
$M^*$. All of those alternating paths and cycles are either augmenting
paths/cycles, or their $M$-edges have exactly the same weight as their
$M^*$-edges. The total gain of all the (vertex-disjoint) augmenting paths and cycles
induced by $F$ is therefore exactly $w(M^*)-w(M)$.

Dealing with augmenting paths and cycles in the context of
edge-weighted graphs is much more challenging than in the unweighted
case. Every augmenting path in unweighted graphs improves the matching
by exactly one and augmenting cycles do not exist. Further, the
classic results of Hopcraft and Karp \cite{HopcroftK73} imply that if
the shortest augmenting path is of length $\ell$, augmenting over an
augmenting path of length $\ell$ cannot create new augmenting paths of
length $\leq\ell$ and after augmenting over a maximal set of
vertex-disjoint augmenting paths of length $\ell$, one gets a
matching with a shortest augmenting path length 
$\geq \ell+2$. If in some matching $M$, the shortest
augmenting path has length $\ell = 2k - 1 \geq 1$, we further know that $M$ already is
a $(1 - \frac{1}{k})$-approximation of the optimal matching. The
$(1-\eps)$-approximation algorithm for unweighted matching in
\cite{lotker15} heavily relies on all those properties.

Some of the properties of augmenting paths for unweighted matchings
also carry over to the weighted case. We will later show that there exists
a set of vertex-disjoint augmenting paths and cycles of length at most
$\ell$ for some $\ell=O(1/\eps)$ such that augmenting over all those
paths/cycles improves the current matching by at least
$\frac{\eps}{4} \cdot w(M^*)$ (see \Cref{lemma:augPath}). Basically,
the existence of this collection of short augmenting paths can be
proven by breaking long augmenting paths and cycles in the symmetric
difference $F=M\triangle M^*$ into short augmenting paths. However,
while in the unweighted case, a large set of vertex-disjoint short
augmenting paths can be computed efficiently in the \CONGEST model
(see \cite{yehuda17,lotker15}), it is not clear how to efficiently
compute such a set of augmenting paths/cycles for weighted matchings
in the \CONGEST model, even in bipartite graphs (there are efficient
algorithms in the \LOCAL model and this has also been exploited in the
literature, e.g., in
\cite{stoc18_edgecoloring,lotker15,nieberg08}). Fortunately, there
still is an efficient \CONGEST algorithm for computing a
$(1-\eps)$-approximate weighted matching in bipartite
graphs~\cite{ahmadi18} (see also Lemma
\ref{lemma:ApproxAndRound}). Unlike the existing \CONGEST algorithms
that are based on the Hopkroft/Karp framework, the algorithm of
\cite{ahmadi18} is even deterministic. It is however not based on
augmenting along short augmenting paths or cycles. Instead, it is
based on linear programming and on a deterministic rounding scheme
that was introduced in \cite{rounding}. As a result, the matching
computed by the algorithm of \cite{ahmadi18} does not have the nice
structural properties of the matchings computed by algorithms based on
the Hopkraft/Karp framework (such as not having any short augmenting
paths or cycles).

\para{Our randomized weighted matching algorithm.} The general idea of our approach is as simple as the algorithm
for the unweighted case in \cite{lotker15}.\footnote{Our
	algorithm is essentially the same one as the one in
	\cite{lotker15,yehuda17}. We just replace the bipartite matching algorithm
	used as a subroutine. The analysis is then however
	different from the analysis in \cite{lotker15,yehuda17}.} We first describe
the randomized version of our algorithm. We start with an
initial matching $M_0$ of the given weighted graph $G=(V,E,w)$ and it
then consists of iterations $i=1,2,\dots$. In each
iteration, we update the given matching such that at the end of
iteration $i$, we have matching $M_i$. In each iteration $i$, we
sample a bipartite subgraph $H_i=(\hat{V}_i, \hat{E}_i)$ of $G$ as follows. Every
node $v\in V$ colors itself black or white independently with
probability $1/2$. An edge is called \textit{monochromatic} if both of
its endpoints have the same color, otherwise, we call the edge
\textit{bichromatic}. To preserve good intermediate results, we keep
all the matching edges of the previous matching not occurring in the
bipartition, i.e., we keep monochromatic matching edges. We call a
node \textit{free} regarding to matching $M$ if none of its incident
edges are $\in M$.

\smallskip

\begin{center}
	\begin{minipage}{1.0\linewidth}
		\vspace{-8pt}
		\begin{mdframed}[hidealllines=false, backgroundcolor=gray!10]
			\textbf{Construct bipartite subgraph
				\boldmath$H_i = (\hat{V}_i, \hat{E}_i)$ of
				$G$ based on matching $M_{i-1}$:}
			\vspace*{-8pt}
			
			\begin{enumerate}
				\item Color each node black or white.
				\item $\hat{V}_i := \{u \in V ~ | ~ u \text{ is free or } \exists \{u,v\} \in M_{i-1} \text{ s.t. } \{u,v\} \text{ is bichromatic}  \}$.
				\item $\hat{E}_i := \{ \{u, v\} \in E ~ | ~ u, v \in \hat{V}_i \text{ and } \{u, v\} \text{ is bichromatic} \}$.
			\end{enumerate}
		\end{mdframed}
	\end{minipage}
	\vspace{-8pt}
\end{center}

After sampling the bipartite subgraph $H_i$, we use the algorithm of
\cite{ahmadi18} to compute a $(1-\lambda)$-approximate weighted
matching of $H_i$ for a sufficiently small parameter $\lambda>0$
($\lambda$ will be exponentially small in $1/\eps$). We
then update the existing matching $M_{i-1}$ by replacing the
bichromatic matching edges with the matching edges of the newly
computed matching of $H_i$. Because there is a collection of short
augmenting paths and cycles with total gain
$\Theta(w(M^*)-w(M_{i-1}))$ (where $M^*$ is an optimal matching of
$G$), we have a reasonable chance of sampling such paths and cycles so
that the matching on $H_i$ can potentially be improved by a
sufficiently large amount. Note however that our algorithm does not
guarantee that the quality of the matching improves monotonically
during the algorithm. However, because the algorithm of
\cite{ahmadi18} has deterministic guarantees if we choose $\lambda$
sufficiently small, we have the guarantee that $w(M_i)$ cannot be much
worse than $w(M_{i-1})$. With the right choice of parameters, it turns
out that $2^{O(1/\eps)}$ iterations are sufficient to obtain a
$(1-\eps)$-approximate matching with at least constant
probability. The details appear in \Cref{sec:randomizedMWMdetails}.

\para{Our deterministic weighted matching algorithm.} The basic idea
of our deterministic algorithm is the same as for the randomized
algorithm. However, we now have to compute the bipartition into black
and white nodes in each iteration deterministically. For this purpose,
we show in \Cref{lemma:collections} that for some
$T=2^{O(1/\eps)} \cdot \ln n$, there exist a collection of
bipartitions $H_1, \hdots, H_T$ such that every path/cycle of length
at most $O(1/\eps)$ (and thus also every augmenting path/cycle of this
length) of $G$ appears in at least one of these bipartitions. Of
course, after changing the matching, also the set of augmenting paths
and cycles changes and one can therefore not just iterate over all
bipartitions $H_1,\dots,H_T$, improve the matching for each
bipartition by using a generic MWM approximation algorithm, and
guarantee that at the end the resulting matching is a sufficiently
good approximation. However, the property of $H_1,\dots,H_T$
guarantees that for a fixed initial matching $M$, when going over
all $T$ bipartitions, there exists one bipartition that can improve
the weight of $M$ by $\Theta(w(M)/T)$. We therefore proceed as
follows. We iterate $O(T)$ times through the sequence
$H_1, \hdots, H_T$ of bipartitions.  For each bipartition, we use the
(deterministic) algorithm of \cite{ahmadi18} to compute a
$(1-\lambda)$-approximate weighted matching of the current bipartite
graph. We however then switch to the new matching if the new matching
improves the old one by a sufficiently large amount. Checking if a
given bipartition leads to a sufficiently large improvement of the
current matching can be done efficiently by first applying the
diameter reduction technique given by \Cref{thm:diameterreduction}.
The details appear in \Cref{sec:deterministicMWMdetails}.

\subsection{Vertex Cover Lower Bound}
\label{sec:loweroverview}

Recall that G\"o\"os and Suomela in \cite{goeoes14_DISTCOMP} showed
that there exists a constant $\eps_0>0$ such that computing a
$(1+\eps_0)$-approximation of minimum (unweighted) vertex cover in
bipartite graphs of maximum degree $3$ requires $\Omega(\log n)$
rounds even in the \LOCAL model. We next describe the high-level
idea of extending this lower bound to show that for $\eps\leq \eps_0$,
computing a $(1+\eps)$-approximate vertex cover in bipartite graphs of
maximum degree $3$ requires $\Omega\big(\frac{\log n}{\eps}\big)$
rounds in the \LOCAL model. Assume that $G$ is the lower bound graph
of the proof of \cite{goeoes14_DISTCOMP}. We transform $G$ into a
bipartite graph $H$ by replacing each edge of $G$ with a path of length
$2k+1$ for some parameter $k=\Theta(1/\eps)$. One can then show that
every vertex cover $S_H$ of $H$ can be transformed (in $O(k)$ rounds
on $H$) into a vertex cover $S_H'$ of $H$ of size $|S_H'|\leq |S_H|$
such that $S_H'$ consists of exactly $k$ inner nodes of every path
replacing a $G$-edge and a vertex cover of $G$ (composed of the  
original $G$-nodes in $H$ that are contained in $S_H'$). By choosing
$k$ appropriately, one can then show that a $(1+\eps)$-approximate vertex
cover on $H$ induces a $(1+\eps_0)$-approximate vertex cover on
$G$. The lower bound then follows because the nodes of $G$ can locally
simulate graph $H$ such that $k=O(1/\eps)$ rounds on $H$ can be
simulated in $O(1)$ rounds on $G$. A detailed proof appears in
\Cref{sec:lower}.


\section{Technical Details: Weighted Vertex Cover Algorithms}
\label{sec:MWVC}

\subsection{Bipartite Weighted Vertex Cover} 
We are now going to provide the additional technical details that are necessary to prove \Cref{thm:approxMWVCbipartite}, which states that in bipartite graphs, one can deterministically compute a $(1+\eps)$-approximation of the MWVC problem in time $\poly\big(\frac{\log n}{\eps}\big)$. As discussed in \Cref{sec:diameter_reduction}, we first use \Cref{thm:diameterreduction} to cluster the graph into clusters of small weak diameter. Further, in \Cref{sec:basicbipartiteMWVC}, we gave a basic distributed MWVC approximation algorithm that can be run efficiently in all clusters. However, to run this algorithm on a weighted bipartite graph, we need a fractional $w$-matching with no short augmenting paths. In \Cref{sec:noshortpaths}, we described the high-level idea that we use to obtain such a fractional $w$-matching. Starting from a $(1-\delta)$-approximate fractional $w$-matching, where $\delta$ needs to be chosen sufficiently small, we reduce the weights of some nodes and the fractional values of some edges to obtain an instance with weights $w'$ and a fractional $\vec{y}'$-matching with no short augmenting paths so that we can then apply the MWVC algorithm of \Cref{sec:basicbipartiteMWVC}. We will now provide the details of this step.

Let us therefore assume that we are given a bipartite node-weighted graph $G=(A\cup B, E, w)$ and that for an appropriate $\delta>0$, we are given a $(1-\delta)$-approximate fractional $w$-matching $\vec{y}$ of $G$. For technical reasons, we further assume that for every $e\in E$, either $y_e=0$ or $y_e>1/n^c$ for some constant $c>0$. Similarly, we assume that for every node $v\in V$, $s(v)=0$ or $s(v)>1/n^c$. We can guarantee this requirement by making sure that every fractional value $y_e$ is an integer multiple of $1/n^c$. As long as $\delta\geq 1/\poly(n)$ (which is much smaller than what we need), this is straightforward to guarantee. As discussed in \Cref{sec:noshortpaths}, for getting rid of short augmenting paths at small cost, we need to find sets $X\subseteq \set{v\in A\cup B : s(v)>0}$ and $F\subseteq \set{e\in E : y_e>0}$ such that all augmenting paths of length at most $2k-1$ (for a given integer $k\geq 1$) are `covered' and $s(X)+y(F)$ is as small as possible. We now first show that we can find sets $X$ and $F$ such that $s(X)+y(F)$ is within a small factor of $s(X^*)+y(F^*)$ for best possible sets $X^*$ and $F^*$. To compute $X$ and $F$, we go through stages $d=1,3,\dots,2k-1$, where in stage $d$ we select sets $X_d$ and $F_d$ that cover all remaining augmenting paths of length $d$. In the end, we set $X=X_1 \cup \dots \cup X_{2k-1}$ and $F=F_1 \cup \dots \cup F_{2k-1}$. Let us now focus on one stage $d$ in the following part.

\subsubsection{Getting rid of augmenting paths of length \boldmath$d$.}

First note that since we covered all augmenting paths of length shorter than $d$ before stage $d$, at the beginning of stage $d$, there are no augmenting paths of length shorter than $d$. To simplify notation, we will in the following use $w$ and $\vec{y}$ to denote the node weights and the fractional edge values at the beginning of stage $d$. We therefore assume that w.r.t.\ $w$ and $\vec{y}$, $G$ has no augmenting paths of length shorter than $d$ and the goal is to find sets $X_d$ and $F_d$ to cover all augmenting paths of length $d$. The problem of finding such sets that minimizes $s(X_d)$ and $y(F_d)$ can be phrased as a minimum weighed set cover problem as follows. The ground set $\mathcal{P}_d$ is the set of all augmenting paths of length $d$ w.r.t.\ the fractional $w$-matching $\vec{y}$ in $G$. We define
\begin{eqnarray*}
  S_d & := & \set{ v\in A\cup B : s(v)>0 \land \text{$v$ is an end node of some augmenting path of length $d$}},\\
  Y_d & := & \set{ e\in E : y_e>0 \land \text{$e$ is an even edge of some augmenting path of length $d$}}.
\end{eqnarray*}
For each $v\in S_d$, let $P_v $ be the set of augmenting paths of
length $d$ for which node $v$ is one of the end nodes. Similarly, for
each $e\in Y_d$, let $P_e$ be the set of augmenting paths of length
$d$ for which it contains $e$ as an even edge. The collection of sets
$\mathcal{S}\in 2^{\mathcal{P}_d}$ is then made up of sets $P_v$ for
each node $v\in S_d $ and $P_e$ for each edge $e\in Y_d$.  The weight
of $P_v$ is $s(v)$ and the weight of $P_e$ is $ y_e$. A solution to
this set cover instance is a collection of sets $P_v$ for $v\in X_d$
and $P_e$ for $e\in F_d$ where $X_d\subseteq S_d $
 and $F_d\subseteq Y_d$,
  and such that $\bigcup_{v\in X_d}P_v\cup \bigcup_{e\in F_d}P_e=\mathcal{P}_d$. The weight of such a set cover is exactly $s(X_d)+ y(F_d)$. Note that since we assumed that all non-zero slack values $s(v)$ and fractional values $y_e$ are at least $1/\poly(n)$ and because we assumed that all node weights are upper bounded by a polynomial in $n$, the weights of all sets in our set cover instance are between $1/\poly(n)$ and $\poly(n)$. In the following, we use $\wmin$ and $\wmax$ to refer to the minimum and the maximum weight of a set in the current set cover instance.

\para{Layered structure of shortest augmenting paths.} To achieve our goal, we first have a look at the structure of the augmenting paths of length $d$ in $G=(A\cup B, E, w)$. For this, we define the following layered graph $H=(V_0\cup\dots\cup V_d, E_H)$. The layer $V_0:=\set{v \in A : s(v)>0}$ consists of the nodes in $A$  with positive slack. For every odd $i$,  layer $V_i$ consists of the nodes in $B\setminus \ \bigcup_{j=0}^{i-1} V_j$ 
 for which there exists an edge from a node in $V_{i-1}$. For every even $i\geq 2$, layer $V_i$ consists of the nodes in $A\setminus \bigcup_{j=0}^{i-1} V_j$ 
for which there exists an edge $e$ with $y_e>0$ from a node in $V_{i-1}$.  Finally, layer $V_d$ consists of nodes in $B$ with positive slack $s(v)$ such that there exists an edge from a node in $V_{d-1}$. Note that if the bipartition of the nodes of $G$ into $A$ and $B$ is known, the layers $V_0\dots,V_d$ can be computed in a simple parallel BFS exploration of $G$ in $d$ rounds. Note also that by definition of the layered structure, every path $p=(v_0,\dots,v_d)$ such that for every $i\in \set{0,\dots,d}$, $v_i\in V_i$, is an augmenting path. Further, because the shortest augmenting path length is $d$, every augmenting path of length $d$ must have this structure and we also know that all nodes $v$ in odd layers $V_i$ for $1<i<d$ have slack $s(v)=0$.

\para{Implementation of greedy weigted set cover algorithm.} We use a
variant of the greedy weighted set cover algorithm here to find the
sets $X_d$ and $F_d$ covering all the shortest augmenting paths in
$G$. We initially set $X_d=F_d=\emptyset$ and we call all nodes and
edges active. Whenever we add a node $v$ to $X_d$ we deactivate $v$
and if we add an edge $e$ to $F_d$, we deactivate $e$. The set of
active augmenting paths of length $d$ is defined as the set of paths
$(v_0,\dots,v_d)$ with $v_i\in V_i$ such that all nodes and edges of
the path are active.  Note that the sets $X_d$ and $F_d$ cover all the
shortest augmenting paths (i.e., they form a solution of the weighted
set cover instance we need to solve) if and only if the number of
active augmenting paths of length $d$ is zero.  The efficiency of a
set $P_v\in \calP_d$ for $v\in A\cup B$ is defined as
$|\set{P\in P_v : P\ \mathrm{active}}|/s(v)$ and the efficiency of a
set $P_e\in \calP_d$ for $e\in E$ is defined as
$|\set{P\in P_e : P\ \mathrm{active}}|/y_e$. Note that the efficiency
of a set $P_v$ or $P_e$ is the number of newly covered paths per
weight when adding $P_v$ or $P_e$ to the set cover (and thus $v$ to
$X_d$ or $e$ to $F_d$). Since every node or edge can be in at most
$\Delta^d$ different paths of length $d$, the minimum possible
efficiency is equal to $1/\wmax$ and the maximum possible efficiency
is equal to $\Delta^d/\wmin$. Computing the efficiencies requires to
count the number of active augmenting paths that pass through each
node and edge. This can be done in $O(d^2)$ rounds in the \CONGEST
model by slightly adapting an algorithm described in
\cite{yehuda17,lotker15}. The details appear in
\Cref{lemma:pathcounting} in \Cref{sec:pathcounting}. The algorithm
consists of phases $i=1,2,\dots,\big\lceil\log\big(\Delta^d\cdot\frac{\wmax}{\wmin}\big)\big\rceil$,
where in phase $i$ we add sets $P_v$ and $P_e$ with efficiency
$\geq 2^{-i}\cdot \frac{\Delta^d}{\wmin}$ to the set cover:

\medskip

\begin{center}
  \begin{minipage}{1.0\linewidth}
    \vspace{-8pt}
    \begin{mdframed}[hidealllines=false, backgroundcolor=gray!10]
      \textbf{Covering Paths of Length \boldmath$d$: Phase $i\in\set{1,\dots,\big\lceil\log\big(\Delta^d\cdot\frac{\wmax}{\wmin}\big)\big\rceil}$}\\[1mm]
      First go over level $\ell=0$ and then sequentially iterate over all odd levels $\ell=1,3,\dots,d$:\\
      Initially all nodes and edges are active.
      \vspace*{-2mm}
      \begin{enumerate}
      \item Count the number of active augmenting paths of length $d$ passing through each node and edge (by using \Cref{lemma:pathcounting}).
      \item If $\ell\in\set{0,d}$, for each active node $v \in V_{\ell}$ do
        \begin{itemize}
        \item Compute efficiency of set $P_v$ (note that $s(v)>0$ because $v\in V_0\cup V_{d}$)
        \item If efficiency of $P_v$ is $\geq 2^{-i}\cdot\frac{\Delta^d}{\wmin}$, add $v$ to $X_d$ and deactivate $v$,
        \end{itemize}
      \item If $\ell\in\set{1,\dots,d-2}$, for each edge $e=\set{u,v}$ with $u\in V_{\ell}$ and $v\in V_{\ell+1}$ do
        \begin{itemize}
        \item Compute efficiency of set $P_e$ (note that $y_e>0$ by construction of layers)
        \item If efficiency of $P_e$ is $\geq 2^{-i}\cdot\frac{\Delta^d}{\wmin}$, add $e$ to $F_d$ and deactivate $e$.
        \end{itemize}
      \end{enumerate}
    \end{mdframed}
  \end{minipage}
  \vspace{-8pt}
\end{center}

By iterating over the layers $V_\ell$, for $\ell= 0,..., d$, we can
add sets $P_v$ (or $P_e$) for $v$ (or $e$) on the same layer in
parallel (note that such set of paths need to be disjoint and they
therefore do not influence each other). At the end of a phase $i$, it
is guaranteed that there is no set of efficiency
$\geq 2^{-i}\cdot \frac{\Delta^d}{\wmin}$ left. The algorithm
therefore always adds sets with an efficiency that is within a factor
$2$ of the largest current efficiency. The details of the algorithm
for phase $i$ is given in the following.

Next, we show that the above algorithm can be implemented efficiently in the \CONGEST model and that the total weight $s(X_d)+y(F_d)$ of the solution output by the algorithm in stage $d$ is small compared to the weight of an optimal vertex cover.

\begin{lemma}\label{lemma:singlestage}
  Let $\vec{y^*}$ be an optimal fractional $w$-matching of the graph $G=(A\cup B, E, w)$ in the current stage. If the given fractional $w$-matching $\vec{y}$ satisfies $y^*(E)-y(E)\leq \delta \cdot w(S^*)$ for an optimal weighted vertex cover of $G$, then the deterministic algorithm above  finds a set $X_d \subseteq S_d$
  and a set $F_d \subseteq Y_d$
  such that $s( X_d) +y(F_d) \leq \alpha_d \delta \cdot w(S^*)$, where $\alpha_d= (d+3)(1+d \ln \Delta)$. The time complexity of the algorithm is $O(d^4\cdot\log n)$ in the \CONGEST model.
\end{lemma}
\begin{proof}	
  First, we look at the time complexity. The algorithm consists of $O(\log (\Delta^d \wmax/\wmin))=O(d\log n)$ phases, where in each phase we iterate over the $O(d)$ levels and for each level the time complexity is upper bounded by computing the best efficiency, which essentially comes down to finding the  number of active augmenting paths passing through each node or edge. By \Cref{lemma:pathcounting}, the number augmenting paths passing through each node and edge can be computed in $O(d^2)$ \CONGEST rounds. The resulting time complexity is therefore $O(d^4\cdot \log n)$ rounds in the \CONGEST model.
  	
  Next, we show that our distributed algorithm is equivalent to a
  variant of the sequential greedy weighted set cover algorithm and
  compute its approximation ratio. Recall that in the standard
  sequential greedy weighted set cover algorithm, one starts with an
  empty set cover and iteratively adds to it a set with the current
  maximum efficiency. It is also well known that the approximation
  ratio of this greedy algorithm is at most $H(q)\leq 1+\ln q$, where
  $q$ is the maximum set size and $H(q)$ denotes the $q^{\mathit{th}}$
  harmonic number~\cite{chvatal79}. Further, it is easy to see that if
  we relax each greedy step to pick a set of at least half the current
  maximum efficiency, then the algorithm will have an approximation
  ratio that is at most $2H(q)$ (this is the greedy set cover
  algorithm variant that we are considering here). To see
  this, notice that a set of at least half the maximum efficiency can
  be turned into a set of maximum efficiency by reducing its weight
  by a factor of at most $2$. The resulting solution is therefore an $H(q)$-approximation of a set cover instance in which all weights are divided by a factor between $1$ and $2$.

  Returning now to the distributed set cover algorithm for this lemma,
  notice that we sequentially iterate over level 0 and the odd levels
  $1, 3, ..., d$ and in each iteration, we only add sets $P_v $(or
  $P_e$) on the same level whose efficiency is at least half the
  current maximum efficiency. Moreover note that sets $P_v$ for nodes
  $v$ in one particular level $\ell$ and sets $P_e$ for edges $e$ in
  one particular level $\ell$ cover disjoint sets of paths. Hence, in
  each parallel step, our algorithm always adds sets covering a
  disjoint set of paths (and therefore a disjoint set of elements in
  the set cover instance). Our distributed algorithm is therefore
  equivalent to an execution of the described variant of the
  sequential greedy algorithm\footnote{Note that in each iteration, in
    step 1, we also recompute the number of paths passing through each
    node and edge so that we always only add sets that cover a
    sufficient number of still uncovered paths.}. Therefore our
  algorithm has an approximation ratio of at most
  $2H({\Delta^d})\leq 2\big(1 + \ln \Delta^d\big)=2(1+d\ln\Delta)$.
  Note that in our setting, the maximum set size is at most $\Delta^d$
  because $|P_v|\leq \Delta^d$ and $|P_e| \leq \Delta^d$.  To upper
  bound the absolute cost $s(X_d)+y(F_d)$ of the computed solution, we
  also need to understand the cost of an optimal set cover solution.

  \para{Interpreting as minimum weighted hypergraph vertex cover.} 
  To better understand our set cover instance, we interpret it as an instance of a weighted minimum vertex cover problem on hypergraphs. We define a weighted hypergraph $H_d=(V_{H},E_{H},\omega)$ as follows. The vertex set $V_{H}$ consists of the sets $P_v$ with $v \in S_d$ and the sets $P_e$ with $e \in Y_d$.
 That is, we have a node in $V_H$ for each $v\in S_d$ and for each $e\in Y_d$. The hyperedges $E_{H}$ of $H_d$ are the augmenting paths of length $d$. A node $P_v\in V_H$ is incident to the hyperedge defined by an augmenting path $p$ if $p\in P_v$, i.e., $p$ starts at node $v$. A node $P_e\in V_H$ is incident to the hyperedge defined by an augmenting path $p$ if $p\in P_e$, i.e., if $e$ is contained in $p$. The weights of vertices $P_v$ and $P_e$  of $H_d$ are defined as $\omega(P_v)=s(v)$ and $\omega(P_e)=y(e)$. Observe that the weighted vertex cover on hypergraph $H_d$ is exactly equivalent to our weighted set cover instance. The dual problem of the natural LP relaxation of weighted vertex cover on hypergraphs is a natural extension of the fractional $w$-matching problem defined for graphs in \Cref{sec:defProblems}. For our hypergraph $H_d$ with weight function $\omega$, we can define a fractional $\omega$-matching of $H_d$ as an assignment $\vec{z}$ of non-negative fractional values $z_{e_p}>0$ for every hyperedge $e_p\in E_H$ such that for every node $x\in V_H$, the sum of the fractional values $z_{e_p}$ of its incident hyperedges sums up to at most $\omega(x)$. A fractional assignment is said to be maximal if for each hyperedge $e_p$, there exists a vertex $x \in e_p$ that is saturated, i.e., for which the fractional values of its hyperedges sum up to $\omega(x)$. Note that if we are given a maximal $\omega$-matching $\vec{z}$ of a hypergraph with node weights $\omega$, the set of saturated vertices w.r.t.\ $\vec{z}$ form a vertex cover of the hypergraph. If the rank of the hypergraph is $R$, every hyperedge can have at most $R$ saturated vertices and the total weight of the saturated nodes is therefore at most $R$ times the total sum of all fractional hyperedge values. Because every augmenting path of length $d$ is in exactly $2$ of the sets $P_v$ and in exactly $(d-1)/2$ of the sets $P_e$, the rank of our hypergraph $H_d$ is equal to $2+(d-1)/2=(d+3)/2$. The weight of an optimal vertex cover of the hypergraph $H_d=(V_H,E_H,\omega)$ is therefore at most $\frac{d+3}{2}\cdot z(E_H)$, where $z(E_H)=\sum_{e_p\in E_H} z_{e_p}$ if $\vec{z}$ is a maximal fractional $\omega$-matching of $H_d$.

  \para{Upper bounding the value of a fractional \boldmath$\omega$-matching of $H_d$.} In order to relate the total value $z(E_H)$ of an $\omega$-matching $\vec{z}$ of $H_d$ to the cost of an optimal weighted vertex cover on our original graph $G$, we interpret the fractional assignment $\vec{z}$ on graph $G$. Since the hyperedges in $H_d$ correspond to the augmenting paths of length $d$ of the graph $G$, the fractional $\omega$-matching $\vec{z}$ of $H_d$ assigns a fractional value $z_p>0$ to each such augmenting path $p$ of $G$. Because $\vec{z}$ is a valid fractional $\omega$-matching of $H_d$, we know that for every node $v\in V_0\cup V_d$ and for every edge $e\in E$ between two levels $V_{2i-1}$ and $V_{2i}$, the sum of the $z_p$-values for all augmenting paths $p$ that include $v$ is at most $s(v)$ and the sum of the $z_p$-values for all augmenting paths $p$ that include $e$ is at most $y_e$. This however implies that we can in parallel augment our existing fractional $w$-matching of $G$ on all augmenting paths $p$ as follows For every edge $e$ of $G$, let $\calP_d(e)\subseteq \calP_d$ be the set of augmenting paths of length $d$ that contain edge $e$. For every edge $e$ between an even level $V_{2i}$ and a consecutive odd level $V_{2i+1}$, we set $y_e':=y_e + \sum_{p \in\calP_d(e)} z_p$ and for every edge $e$ between an odd level $V_{2i-1}$ and a consecutive even level $V_{2i}$, we set $y_e':=y_e - \sum_{p\in \calP_d(e)} z_p$. By construction, the resulting fractional assignment $\vec{y}'$ is still a valid $w$-matching of $G$ and we have $y'(E)-y(E)=\sum_{p\in \calP_d} z_p = z(E_H)$. By the assumptions of the lemma, we further know that $y'(E)-y(E)\leq y^*(E)-y(E) \leq \delta w(S^*)$.
   
   Let $X_d^*$ and $F_d^*$ be the node and edge sets induced by an optimal solution of the weighted set cover instance. Combining everything, we therefore obtain 
  \begin{eqnarray*}
    s(X_d)+y(F_d) & \leq & 2(1+d\ln \Delta)\cdot (s(X_d^*)+y(F_d^*))\\
    & \leq & 2(1+d\ln\Delta)\cdot \frac{d+3}{2} \cdot z(E_H)\\
    & \leq & (d+3)(1+d\ln\Delta)\cdot\delta\cdot w(S^*).
  \end{eqnarray*}
This concludes the proof of the lemma.
\end{proof}

We now have everything that we need to prove our first main theorem, which shows that a $(1+\eps)$-approximation of the minimum weighted vertex cover problem can be computed deterministically in $\poly\big(\frac{\log n}{\eps}\big)$ rounds in the \CONGEST model.

\begin{proof}[\textsf{\textbf{Proof of \Cref{thm:approxMWVCbipartite}}}]
  By \Cref{thm:diameterreduction}, it is sufficient to show that we can compute a $(1+\eps)$-approximate weighted vertex cover solution in time $\poly\big(\frac{\log n}{\eps}\big)\cdot D$ in bipartite graphs in the \SUPPORTED model with a bipartite communication graph of diameter $D$. Let us therefore assume that we are given a node-weighted bipartite graph $G=(A\cup B, E, w)$ on which we want to compute a $(1+\eps)$-approximate weighted vertex cover in the \SUPPORTED model with a bipartite communication graph $H$ of diameter $D$. Note that because $H$ has diameter $D$, we can compute a $2$-coloring of $H$ in $O(D)$ rounds in the \CONGEST model and because $G$ is a subgraph of $H$, we can therefore also compute the bipartition of the nodes of $G$ into $A$ and $B$ in time $O(D)$. In the following, we will therefore assume that the nodes of $G$ know if they are in $A$ or in $B$.

  As described in \Cref{sec:noshortpaths}, as a first step for solving MWVC in $G$, we choose a sufficiently small parameter $\delta>0$ and we compute a $(1-\delta)$-approximate fractional $w$-matching $\vec{y}$ of $G$. By using \Cref{thm:fractional_wmatching}, we can do this in time $\poly\big(\frac{\log n}{\delta}\big)$ (recall that we assume that all weights are polynomially bounded integers). As discussed in \Cref{sec:noshortpaths}, we now need to transform our instance so that we have no short augmenting paths and can afterwards apply \Cref{lemma:basicMWVCalg} to compute a vertex cover of $G$. To transform the instance, we need to determine sets $X\subseteq \bigcup_{d=1}^{2k-1} S_d$ and $F\subseteq \bigcup_{d=1}^{2k-1} Y_d $ such that for every augmenting path $P=(v_0,v_1,\dots,v_d)$ of length $d\leq 2k-1$ for $k=\lceil 2/\eps\rceil$, either $v_0$ or $v_d$ is in $X$ or one of the edges $\set{v_{2i-1},v_{2i}}$ for $i\in \set{1,\dots,(d-1)/2}$ is in $F$. If we have such sets $X$ and $F$, then compute new node weights $w'$ and a new fractional $w'$-matching $\vec{y}'$ as given by \Cref{eq:conversion}. \Cref{lemma:conversion} implies that w.r.t. weights $w'$ and the fractional $w'$-matching $\vec{y}'$, graph $G$ then has no augmenting paths of length at most $2k-1$. Further, \Cref{lemma:approxpreserved} shows that if we compute an $(1+\eps/2)$-approximate weighted vertex cover $S$ of $G$ for weights $w'$, then $w(S)\leq (1+\eps/2)\cdot w(S^*) + s(X) + y(F)$, where $S^*$ is an optimal weighted vertex cover of $G$ w.r.t.\ the original weights $w$. By \Cref{lemma:basicMWVCalg}, the $(1+\eps/2)$-approximate vertex cover of $G$ w.r.t.\ weights $w'$ can be compute in time $O(D+1/\eps)$ in the \SUPPORTED model with a communication graph of diameter $D$.
  
  We therefore need to show that we can find appropriate sets $X$ and $F$ such that $s(X)+y(F)\leq \frac{\eps}{2}\cdot w(S^*)$. We compute the sets $X$ and $F$ in stages $d=1,3,\dots$ by iterating over the possible augmenting path lengths and by using \Cref{lemma:singlestage}. After each augmenting path length, we will apply the conversion given by \eqref{eq:conversion}. For each stage $d=1,3,\dots,2k-1$, we define $w_d$ to be the weight function that is used in stage $d$ and $\vec{y_d}$ to be the fractional $w_d$-matching that is used in stage $d$. For the first stage, we therefore have $w_1=w$ and $\vec{y_1}=\vec{y}$.  For each stage $d$, we further define $\vec{y_d^*}$ to be an optimal fractional $w_d$-matching. Note that because $\vec{y_1}$ is a $(1-\delta)$-approximate $w$-matching of $G=(V,E,w)$, we have $y_1^*(E)-y_1(E)\leq \delta y_1^*(E)\leq \delta w(S^*)$, where the last inequality follows from \Cref{lemma:duality}. We can therefore apply \Cref{lemma:singlestage} for $d=1$ and we obtain sets $X_1$ and $F_1$. We now obtain new weights $w_3$ and a new fractional $w_3$-matching  $\vec{y_3}$ by applying  \eqref{eq:conversion}. That is, we set $w_3(v)=w_1(v)-s(v)$ for every $v\in X_1$, we set $y_{3,e}=0$ for every $e\in F_1$, and we reduce the weights of the incident nodes of such edges accordingly (for each reduced edge value, the corresponding node weight is reduced by the same amount). Note that whenever we reduce the weights of both nodes of an edge $e$ by the same amount $x$, then the value of an optimal fractional $w$-matching is reduced by at least $x$. We therefore have $y_3^*(E)\leq y_1^*(E)-y_1(F_1)$. By construction, we also have $y_3(E)=y_1(E)-y_1(F_1)$ and we therefore get $y_3^*(E)-y_3(E)\leq \delta\cdot w(S^*)$. We can therefore again apply \Cref{lemma:singlestage} for $d=3$. If we continue like this, we obtain sets $X_1,\dots,X_{2k-1}$ and set $F_1,\dots,F_{2k-1}$ such that for all $d$, $s(X_d)+y(F_d)\leq (d+3)(1+d\ln\Delta)\cdot\delta\cdot w(S^*)\in O(k^2\log\Delta)\cdot \delta\cdot w(S^*)$. By construction, the sets $X:= X_1\cup X_3\cup\dots\cup X_{2k-1}$ and $F:=F_1\cup F_3\cup \dots\cup F_{2k-1}$ cover all augmenting paths of length at most $2k-1$ w.r.t.\ to the original weight function $w$ and the original fractional $w$-matching $\vec{y}$. Because there are $k$ stages, we further have
\[
  s(X)+y(F)\in O(k^3\log \Delta)\cdot \delta\cdot w(S^*) \in O\left(\frac{\log\Delta}{\eps^3}\right)\cdot\delta \cdot w(S^*).
\]
If we choose $\delta\leq c\cdot \eps^4$ for a sufficiently small constant $c>0$, 
we therefore have $s(X)+y(F)\leq \frac{\eps}{2}\cdot w(S^*)$. With this choice for $\delta$, both the algorithm to reduce the diameter (cf.\ \Cref{thm:diameterreduction}) and the algorithm for getting rid of short augmenting paths (cf.\ \Cref{lemma:singlestage}) have a round complexity of $\poly\big(\frac{\log n}{\eps}\big)$, which concludes the proof of the theorem.
\end{proof}

\subsection{Weighted Vertex Cover in General Graphs}
\label{sec:generaldetails}

We next discuss the additional technical details needed to prove our upper bound for approximating minimum (weighted) vertex cover in general graphs. As discussed in \Cref{sec:overviewMWVCgeneral}, the high-level idea is as follows. We first (virtually) construct the bipartite double cover $G_2$ of our graph $G$ and we then apply the $(1+\eps)$-approximation algorithm for bipartite graph from \Cref{thm:approxMWVCbipartite}. We can then transform this solution to a half-integral $(1+\eps)$-approximation of the fractional weighted vertex cover problem on $G$. Removing an independent set from the graph induced by the half-integral nodes gives the desired approximation. The details are given in the following proof of \Cref{thm:approxMWVCgeneral}.

\begin{proof}[\textsf{\textbf{Proof of \Cref{thm:approxMWVCgeneral}}}]
  As discussed, given a weighted graph $G=(V,E,w)$, we first construct the bipartite double cover $G_2=(V_2,E_2,w)$, where the weight function $w$ is extended to $G_2$ in the obvious way (and as described in \Cref{sec:overviewMWVCgeneral}). Note that since every node of $G$ only needs to simulate $2$ nodes in $G_2$ and every edge of $G$ is replaced by $2$ edges in $G_2$, \CONGEST algorithms on $G_2$ can be simulated on $G$ with only constant overhead. By using \Cref{thm:approxMWVCbipartite}, we can therefore compute a $(1+\eps)$-approximation of MWVC on $G_2$ in time $\poly\big(\frac{\log n}{\eps}\big)$. By \Cref{lemma:doublecoverVC}, this can directly be turned into a half-integral $(1+\eps)$-approximate fractional weighted vertex cover of $G$. Let $S_1$ be the nodes of $G$ that have a fractional vertex cover value of $1$ and let $S_{1/2}$ be the nodes of $G$ that have a fractional vertex cover value of $1/2$. Assume further that $I_{1/2}$ is an independent set of the induced subgraph $G[S_{1/2}]$. Clearly, $S:=S_1\cup S_{1/2}\setminus I_{1/2}$ is a vertex cover of $G$. If the weight $w(I_{1/2})$ is $w(I_{1/2})\geq\lambda\cdot w(S_{1/2})$, then the weight of the vertex cover $S$ can be bounded as
  \begin{eqnarray*}
    w(S) & = & w(S_1) + w(S_{1/2}) - w(I_{1/2})\\
         & \leq & w(S_1) + (1-\lambda)\cdot w(S_{1/2})\\
         & \leq & (2-2\lambda)\cdot \left(w(S_1) + \frac{1}{2}\cdot w(S_{1/2})\right)\\
         & \leq & (2-2\lambda)\cdot (1+\eps)\cdot w(S^*),
  \end{eqnarray*}
  where $S^*$ is an optimal weighted vertex cover of $G$. The claim of the theorem now follows.
\end{proof}

We can now also directly prove \Cref{cor:approxMWVCcoloring}, which states that in graphs that can efficiently be colored with $C$ colors, we can compute a $(2-2/C+\eps)$-approximation of MWVC in polylogarithmic time.

\begin{proof}[\textsf{\textbf{Proof of \Cref{cor:approxMWVCcoloring}}}]
  By \Cref{thm:diameterreduction}, we can first reduce the diameter of the communication graph in time $\poly\big(\frac{\log n}{\eps}\big)$ while only losing a $(1+\eps/2)$-factor in the approximation. Now assume that we are given a weighted graph $G$ with a $C$-coloring, that we have a communication graph $H$ of diameter $D$, and that we can use the \SUPPORTED model. In time $O(D)$ in the \SUPPORTED model, we can then use the $C$-coloring to compute an independent set of weight at least $w(V(G))/C$ by just keeping the heaviest color class. The corollary therefore follows directly by combining \Cref{thm:diameterreduction} and \Cref{thm:approxMWVCgeneral}.
\end{proof}


\section{Technical Details: Weighted Matching Algorithms}
\label{sec:MWM}

Before we start with analyzing our matching algorithms, we first provide two helpful lemmas, which we will use in the randomized and deterministic approach. The first lemma states that if a matching $M$ is sufficiently \textit{far} from optimal, then augmenting along a carefully chosen set of vertex-disjoint augmenting paths and cycles up to some maximal length guarantees a gain that is half as large as the maximum possible gain. Similar statements appeared before and were used in parallel algorithms and in the distributed context in the \LOCAL model (e.g., \cite{drake03,stoc18_edgecoloring,HougardyV06,lotker15,nieberg08}). Throughout this section, we assume that we are given an edge-weighted $n$-node graph $G=(V,E,w)$ with maximum degree $\Delta$ and edge weights $w(e)\in \set{1,\dots,W}$ (for $W\leq n^{O(1)}$), and we assume that $M^*$ is an optimal weighted matching of $G$.

\begin{lemma}\label{lemma:augPath}
  For every matching $M$ with $w(M) < (1-\eps/2) \cdot w(M^*)$, there exists a set of vertex-disjoint augmenting paths and cycles of length at most $\ell = O(1/\eps)$ such that augmenting over all those structures increases the weight of $M$ by at least $\frac{\eps}{4} \cdot w(M^*)$.  
\end{lemma}
\begin{proof}
  Let $F:=M\triangle M^*$ be the symmetric difference between $M$ and some optimal matching $M^*$. Clearly, $F$ induces a collection of paths and even cycles. By optimality of $M^*$, each of those paths/cycles either is an augmenting path/cycle or the total weight of all $M$-edges on the path/cycle is equal to the total weight of the total weight of the $M^*$-edges on the path/cycle. The set of augmenting paths and cycles induced by $F$ are vertex-disjoint and they together have a gain of exactly $w(M^*)-w(M)\geq \eps/2 \cdot w(M^*)$. Some of the augmenting paths and cycles might however be long (we could potentially have a single augmenting path or cycle of length $|M|+|M^*|$) and we need to split into short vertex-disjoint augmenting paths and cycles.

	To do this,  we first (arbitrarily) partition each of the long augmenting paths and cycles into subpaths such that all except one subpath are of length exactly $x:=\lceil \ell/3\rceil$ and the last subpath is of length between $x$ and $2x-1$. In all of those subpaths, we define a \textit{light edge}, where the light edge is defined as an $e \in M$ with the lowest weight among all edges in $M$ in this subpath. We now remove the light edge of every subpath. Because we remove one edge of each of the subpaths, this splits each long augmenting path or cycle into short paths of length at most $(x-1)+(2x-1-1)=3(x-1)\leq \ell$.
  
	Since every light edge $e$ is chosen to be minimal among the edges of $M$ and there are at least $x/2-1\geq \ell/6 - 1$ edges of $M$ in each subpath, the total weight of the light edges is at most a $\frac{6}{\ell-6}$-fraction of the weight of all edges in $M$. Summing this up over all subpath and choosing $\ell\geq 24/\eps + 6$, we can lower bound the total gain of the remaining augmenting paths and cycles of length at most $\ell$ by
	\begin{align*}
          w(M^*) - w(M) - \left(  \frac{6}{\ell - 6} \right) w(M) &\geq w(M^*) - \left( 1 - \frac{\eps}{2} \right) \cdot w(M^*) -  \frac{\eps}{4} \cdot w(M) \\
                                                                  &\geq\frac{\eps}{4} \cdot w(M^*)\qedhere
	\end{align*}  
\end{proof}

The second lemma was proven in \cite{ahmadi18} and it in particular shows the existence of an efficient deterministic distributed approximation scheme for maximum weighted matchings in bipartite graphs. We use this algorithm as a subroutine in our randomized and in our deterministic weighted matching algorithm.

\begin{lemma}\label{lemma:ApproxAndRound}\cite{ahmadi18}
  There is a deterministic \CONGEST algorithm to compute a $(1-\eps)$-approximate MWM in bipartite graphs in time $O\left( \frac{\log (\Delta \cdot W)}{\eps^2} + \frac{\log^2(\Delta / \eps) + \log^*(n)}{\eps}\right)$. Further, in the same asymptotic time a $(2/3 - \eps)$-approximation can be computed on general graphs.
\end{lemma}
\begin{proof}
	Combining Theorem 2 and 3 of \cite{ahmadi18} directly gives this result. 
\end{proof}

\subsection{Randomized Matching}
\label{sec:randomizedMWMdetails}

The randomized version of our matching algorithm works as follows: Initially, a $(1/2)$-approximation of the MWM problem is computed on the whole graph $G$. This can be done using \Cref{lemma:ApproxAndRound} (by setting $\eps=1/6$). Let the resulting matching of $G$ be $M_0$. Starting from this matching, we perform $T$ (the value of $T$ will be determined later) iterations of the following procedure. Consider iteration $i\in \set{1,\dots,T}$. In iteration $i$, we transform an existing matching $M_{i-1}$ into a new matching $M_i$. To achieve this, each node independently colors itself black or white with probability $1/2$. Using that coloring, we construct a bipartite subgraph $H_i=(\hat{V}_i, \hat{E}_i)$ by using the construction described in \Cref{sec:matching}. We then use \Cref{lemma:ApproxAndRound} to find a $(1-\lambda)$-approximate weighted matching of the sampled bipartite graph $H_i$, where the parameter $\lambda$ is set to $\eps/(2T)$. The new matching $M_i$ is then defined as the union of the matching edges of this matching of $H_i$ and the matching edges of $M_{i-1}$ that are located outside of the bipartition, i.e., the monochromatic edges from the previous matching.

In the subsequent analysis, we will show that it is sufficient for any $\delta \in (0, 1/2]$ to choose the number of iterations $T$ (i.e., the number of bipartitions) as $T=2^{O(1/\eps)}\cdot \ln^3(1/\delta)$ to find a $(1- \eps)$-approximation with probability at least $1-\delta$.

When computing the matching on $H_i$, the edges of the existing matching $M_{i-1}$ that are part of the graph $H_i$ are completely ignored. As a result, it is possible that $w(M_i)<w(M_{i-1})$, that is, we cannot guarantee that the matching weight is monotonically increasing during our algorithm. To analyze this potential loss in the matching quality, we fix some optimal weighted matching $M^*$ of $G$ and we take a look at the weight difference between the following matchings: $\hat{M}_i^* := \hat{E}_i \cap M^*$ contains the edges of $M^*$ that are sampled in iteration $i$ and $\hat{M}_i' := \hat{E}_i \cap M_{i-1}$ contains the edges of $M_{i-1}$ that are sampled in iteration $i$.

\begin{lemma}\label{lemma:expImp}
	In any iteration $ i > 0$ in which $w(M_{i-1}) < (1-\eps/2) \cdot w(M^*)$, the expected value of $w(\hat{M}_i^*) - w(\hat{M}_i')$ is at least $2^{-O(1/\varepsilon)} \cdot w(M^*)$. Moreover, there is a positive value $c(\eps)=2^{-O(1/\eps)}$  such that we have $w(\hat{M}_i^*) - w(\hat{M}_i') > c(\eps) \cdot w(M^*)$ with probability at least $c(\eps)/8$.
\end{lemma}        
\begin{proof}
	Let $P$ be an augmenting path or cycle of length $k$ w.r.t.\ $M_{i-1}$ in $G$. If $P$ is an augmenting path, then the probability that $P$ is contained in graph $H_i$, i.e., $\Pr(E(P) \subseteq \hat{E}_i) = 2^{-k}$ (where $E(P)$ denotes the set of edges of $P$). If $P$ is an augmenting cycle of (even) length $k$, then $\Pr(P \subseteq \hat{E}_i) = 2^{-k+1}$. 
	
	By \Cref{lemma:augPath}, we know that there exists a set of vertex-disjoint augmenting paths and cycles of length at most $\ell = O(1/\eps)$ and with total gain at least $\frac{\eps}{4} \cdot w(M^*)$. By the above observation, each of those augmenting paths or cycles is contained in $H_i$ with probability at least $2^{-{\ell}}=2^{-O(1/\eps)}$. The lower bound on the expectation of $w(\hat{M}_i^*) - w(\hat{M}_i')$ then follows by linearity of expectation, i.e., $\E[w(\hat{M}_i^*) - w(\hat{M}_i')] \geq 2^{-O(1/\eps)} \cdot w(M^*)$. From here we can conclude that there exist some value $c(\eps)=2^{-O(1/\eps)}$ s.t.  $\E[w(\hat{M}_i^*) - w(\hat{M}_i')] \geq 2c(\eps) \cdot w(M^*)$. The lower bound on the probability follows because the vertex-disjoint paths and cycles are sampled independently where each path/cycle has a gain of at most $w(M^*)$ and we can therefore apply a Chernoff bound (see \Cref{thm:chernoff_bound}):
	\begin{align*}
		\Pr\left(w(\hat{M}_i^*) - w(\hat{M}_i') \leq c(\eps) \cdot w(M^*) \right) 
		&= \Pr\left(w(\hat{M}_i^*) - w(\hat{M}_i') \leq \left(1- \frac{1}{2}\right)\cdot 2c(\eps) \cdot w(M^*) \right) \\
		&\leq e^{-\frac{1}{8} \cdot \frac{2c(\eps)\cdot w(M^*)}{w(M^*)}} = e^{-\frac{c(\eps)}{4}} \leq 1 - \frac{c(\eps)}{8}
	\end{align*}
In the last step we use that $e^{-x} \leq 1 - \frac{x}{2}$ for $0 \leq x \leq 1$.
\end{proof}

\begin{definition}
  Let $c(\eps)=2^{-O(1/\eps)}$ (and $c(\eps)\leq 1$) be chosen as required by \Cref{lemma:expImp}. We call an iteration $i$ \textbf{successful} if $w(\hat{M}_i^*) - w(\hat{M}_i') > c(\eps) \cdot w(M^*)$ is true.
\end{definition}

Note that by \Cref{lemma:expImp}, as long as $w(M) < (1-\eps/2) \cdot w(M^*)$, an iteration is successful with probability at least $c(\eps)/8$.
Let $T$ be the overall number of iterations and recall that $\lambda$ states the approximation factor used to compute the matching in every bipartition.

\begin{lemma}\label{lemma:IgnGoodApprox}
	Choosing $\lambda =\frac{\eps}{2T}$ guarantees that if the algorithm ever computes a matching $M_i$ where $w(M_i) \geq (1 - \eps/2) \cdot w(M^*)$, there will be no subsequent iteration $j > i$ (and $j \leq T$) such that $w(M_j) < (1 - \eps) \cdot w(M^*) $. 
\end{lemma}
\begin{proof}
	Since a $(1 - \lambda)$-approximation is computed on the bipartite subgraph $H_i$ in every iteration $i$, it is clear that $w(M_{i+1}) \geq w(M_i) - \lambda w(M^*)$. Even assuming a  worst-case scenario, where we lose $\lambda w(M^*) = \frac{\eps}{2T} w(M^*)$ in all $T$ rounds, the stated approximation is achieved:
        \[
          w(M_T) \geq \left(1 - \frac{\eps}{2}\right) w(M^*) - \frac{\eps}{2T} \cdot T \cdot w(M^*) = (1 - \eps) w(M^*). \qedhere
        \]
      \end{proof}

As a result of the previous lemma, we can essentially ignore executions in which we reach $w(M_i) \geq (1 - \eps/2) \cdot w(M^*)$ at some point. In the following, as soon as $w(M_i) \geq (1 - \eps/2) \cdot w(M^*)$, all subsequent iterations are called successful independently of how the matching weight changes. We will now show that for $T = \Theta\left( \frac{1}{c(\eps)^2} \log \left( \frac{1}{\delta} \right) \right)$ iterations, the first part of \Cref{thm:approxMWM} holds.

\begin{lemma}\label{lemma:RandApprox}
	 For every $\eps, \delta\in(0, \frac{1}{2}]$, our procedure computes a $(1- \varepsilon)$-approximation of the maximum weighted matching problem in
	 \begin{align*}
	 		 2^{O(1/\eps)}\cdot \left( \log^2 \Delta + \log( \Delta \cdot W) + \log^*n \right) \cdot \log^3(1/\delta)
	 \end{align*}
	  rounds with probability at least $1-\delta$ in \CONGEST.
\end{lemma}
\begin{proof}
  Let $S$ be the number of successful iterations among all the
  $T := \frac{16}{c(\eps)^2} \ln \left( \frac{1}{\delta} \right)$
  iterations. From the definition of a successful iteration and from
  \Cref{lemma:expImp}, we know that every iteration is successful with
  probability at least $c(\eps)/8$, independently of the what happened
  in previous iterations. We therefore have
  $\E[S] \geq T \cdot \frac{c(\eps)}{8} = 2 \cdot \frac{\ln \left(
      \frac{1}{\delta} \right)}{c(\eps)}$. Because the probability for
  being successful is always at least $c(\eps)/8$, independently of
  previous iterations, the random variable $S$ dominates a random
  variable that is defined as the sum of $T$ independent random
  $\set{0,1}$-random variables, which are all equal to $1$ with
  probability exactly $c(\eps)/8$. Applying a Chernoff bound
  (\Cref{thm:chernoff_bound}) results in the subsequent statement.
  \begin{align*}
    \Pr \left(S \leq  \frac{\ln \left( \frac{1}{\delta} \right)}{c(\eps)} \right)
    = \Pr\left(S \leq \frac{\E[S]}{2}\right)
    \leq e^{\frac{\ln \delta}{4c(\eps)}}
    \leq \delta		
  \end{align*}
  In the following, assume that  $w(M_i) \geq (1 - \eps/2) \cdot w(M^*)$ holds at no time (otherwise, \Cref{lemma:IgnGoodApprox} implies that we obtain at least a $(1-\eps)$-approximation at the end).
  We then now know that with probability at least $1-\delta$ we have more
  than $\frac{1}{c(\eps)} \ln \left( \frac{1}{\delta} \right)$ successful
  iterations, where each successful iteration improves the current
  matching by at least $(1- \lambda) c(\eps) \cdot w(M^*)$ while in
  non-successful iterations we lose at most $\lambda \cdot
  w(M^*)$. Recall that we choose $\lambda =\frac{\eps}{2T}$ (regarding
  \Cref{lemma:IgnGoodApprox}) and $w(M_0) \geq w(M^*)/2$. The overall
  weight after $T$ iterations can be bounded as follows:
  \begin{align*}
    w(M_T) &\geq w(M_0) + S \cdot (1- \lambda) c(\eps) \cdot w(M^*) - (T-S) \lambda \cdot w(M^*)  \\
           &\geq w(M^*)/2+ c(\eps) \cdot S \cdot w(M^*) - \lambda \cdot T \cdot w(M^*).
  \end{align*}
  By the assumption that $\delta \leq 1/2$, this term simplifies to
  $(1-\varepsilon)w(M^*)$ with probability at least $1-\delta$,
  because then $c(\eps) \cdot S > \ln(1/\delta) > 1/2$. This proves
  the stated approximation guarantee.

  Since by \Cref{lemma:ApproxAndRound}, each weighted bipartite matching instance (and therefore each iteration) runs in
  $O\left( \frac{\log (\Delta \cdot W)}{\lambda^2} +
    \frac{\log^2(\Delta / \lambda) + \log^*(n)}{\lambda}\right)$
  rounds in \CONGEST, the runtime for
  all $T$ iterations is as stated.
\end{proof}

\subsection{Deterministic Matching}
\label{sec:deterministicMWMdetails}

As stated in \Cref{sec:matching}, the deterministic version of our weighted matching algorithm differs from the randomized one in a few ways. First, instead of running the algorithm on the whole graph, we use \Cref{thm:diameterreduction} to first decompose the network into  clusters of small (weak) diameter, where afterwards each cluster independently computes a good approximation. The clusters allow us to efficiently compute the exact weight of some matching $M$. In each iteration $i$, we first compute a new matching $M_i$ by applying the deterministic weighted matching algorithm of \cite{ahmadi18}. In each cluster, we then compare $w(M_i)$ to the weight $w(M_{i-1})$ of the previous matching. If $w(M_i) -w(M_{i-1})< \frac{\eps}{8T} \cdot w(M_{i-1})$, we discard the new matching and we simply set $M_i$ to be equal to $M_{i-1}$. Note that the parameter $\lambda$ actually has the same meaning but a different value here than in the randomized approach.  The number of bipartitions used is again denoted by $T$, whereby we  have to work multiple times on the same bipartition here. For that reason the algorithm proceeds in stages. In each stage the algorithm iterates through the bipartitions $H_1, \hdots, H_T$, each time excecuting the above procedure. We will later see that in each stage there exists at least one bipartition $H_i$ where our approximation improves the current matching $M_{i-1}$ by at least $\frac{\eps}{8T} \cdot w(M_{i-1})$. Proceeding through $O(T)$ many stages is then sufficient to compute the desired $(1-\eps)$-approximation. 
In the following, we concentrate on the computation in a single cluster of the decomposition. More specifically, assume that we can use the \SUPPORTED model with a communication graph of diameter $\poly\big(\frac{\log n}{\eps}\big)$. The result can then be lifted to the \CONGEST model in general graphs by applying \Cref{thm:diameterreduction}. Working in the \SUPPORTED model allows us to assume that the nodes are uniquely labeled with IDs from $1$ to $n$ (we can relabel the nodes in such a way in time linear in the diameter of the communication graph).

As discussed in \Cref{sec:matching}, in our deterministic algorithm, we need to compute the bipartitions deterministically. The following lemma shows that there exists a sequence of bipartitions that is `good' in every graph and for every collection of short augmenting paths and cycles.

\begin{lemma}\label{lemma:collections}
	Let $N > 0$ and $k \leq N$ be two integers. There exists a collection of $T = O(k \cdot 2^k \cdot \log N)$ functions $f_1, \hdots, f_T \in [N] \rightarrow \{0, 1\}$ such that for every vector $(x_1, \hdots, x_k) \in [N]^k$ with pairwise disjoint entries, there exists a function $f_i$ in the collection such that $f_i(x_j) = j\ \mathrm{mod}\ 2$ for every $1 \leq j \leq k$.
\end{lemma} 
\begin{proof}
	We will prove the lemma with a probabilistic argument. The probability that a randomly chosen $f_i$ maps some $x_j$ to $0$ respectively $1$ is $1/2$. We say an function $f_i$ \textit{takes care} of $(x_1, \hdots, x_k)$ if $(f_i(x_1), \hdots, f_i(x_k)) \in (0,1,0,1,\hdots)$. The probability that $f_i$ takes care for a specific vector is $2^{-k}$. Further, the probability that no function in a collection of $T$ random function takes care of a fixed vector is $(1-2^{-k})^T \leq e^{-T/2^k}$. Since there are $N^k$ different vectors, the probability that there exits a vector such that no function $f_i$ takes care of it, is at most $N^k \cdot e^{-T/2^k}$. Choosing $T > k \cdot 2^k \cdot \ln N$ pushes this probability below $1$, which shows that there must exist a collection of $T$ functions $f_1,\dots,f_T$ s.t. for every possible vector at least one of those functions will take care. 
\end{proof}

We can use \Cref{lemma:collections} to compute a bipartition of our graph as follows. Assume that the nodes are labeled with unique IDs from $1$ to $N$. When using a particular function $f_i$ to partition the nodes, each node just applies $f_i$ to its ID and gets colored black or white accordingly. Note that as observed above, we can assume that $N\leq n$. Our algorithm now works in consecutive phases, where in each phase, we iterate through all the $T$ bipartitions given by \Cref{lemma:collections}. In each phase, this induces $T$ bipartite graphs $H_1,\dots,H_T$. We choose $k=O(1/\eps)$ in the constructions large enough so that \Cref{lemma:augPath} applies (when setting $\ell$ to $k$). We compute a $(1-\lambda)$-approximate matching on each of these graphs $H_i$ by using \Cref{lemma:ApproxAndRound} and as discussed above, we only keep this matching if it improves the weight of the current matching by at least $\frac{\eps}{8T}\cdot w(M_{i-1})$. The next lemma shows that if $\lambda$ is chosen sufficiently small, there is an iteration in every stage, in which we improve the matching.

\begin{lemma}\label{lemma:exOfAug}
  Assume that $\lambda \leq \eps/(8T)$. For every matching $M$ with $w(M) \leq (1 - \eps) w(M^*)$, there exists one bipartition $H_i$ from the collection $H_1, \hdots, H_T$, such that computing a $(1-\lambda)$-approximate weighted matching on $H_i$ improves $M$ by at least $\frac{\varepsilon}{8T}\cdot w(M)$.  \end{lemma} 
\begin{proof}
  Recall that we choose the parameter $k$ in the construction of \Cref{lemma:collections} sufficiently large so that \Cref{lemma:augPath} guarantees that (w.r.t.\ matching $M$) there exists a collection of vertex-disjoint augmenting paths and cycles of length at most $k$ and with total gain at least $\frac{\eps}{4}\cdot w(M^*)$. The statement of \Cref{lemma:collections} leads to the fact that each of these augmenting paths or cycles appears in at least one of the bipartite graphs $H_1,\dots,H_T$. By the pigeonhole principle, there is a bipartite graph $H_i$ in which we can make a total gain of $\frac{\eps}{4T}\cdot w(M^*)$. If we choose $\lambda\leq \eps/(8T)$, then a $(1-\lambda)$-approximate matching in $H_i$  still guarantees a gain of $\frac{\eps}{8T}\cdot w(M^*) \geq \frac{\eps}{8T}\cdot w(M)$.
\end{proof}

We can now prove the second part of \Cref{thm:approxMWM}.

\begin{lemma}\label{lemma:DetApprox}
	There exists a deterministic \CONGEST algorithm to compute a $(1-\eps)$-approximation to the weighted matching in time
	\begin{align*}
		2^{O(1/\eps)} \cdot \poly\log n.
	\end{align*}
\end{lemma}
\begin{proof}
  We first prove the claim on the approximation ratio. Recall that we start with a $1/2$-approximate weighted matching. Then, by \Cref{lemma:exOfAug}, each stage improves the weight of the given matching $M$ by at least $\frac{\eps}{8T}\cdot w(M) \geq \frac{\eps}{16T}\cdot w(M^*)$. Therefore, $O(T/\eps)$ stages suffice to achieve a final $(1-\eps)$-approximation for MWM.
  
  We now show that $O(T/\eps)$ stages of length $T$ can be implemented in the claimed round complexity. By \Cref{thm:diameterreduction}, it is sufficient to prove the claim of the lemma for algorithms in the \SUPPORTED model with a communication graph of polylogarithmic diameter. The algorithm for each bipartition then clearly requires at most $\poly\big(\frac{\log n}{\lambda}\big)$ rounds (\Cref{lemma:ApproxAndRound}). Since we have $T$ bipartitions in every stage and $O(T/\eps)$ many stages, the total round complexity is bounded by
  \[
    \poly\left(\frac{\log n}{\lambda}\right)\cdot T\cdot O\left(\frac{T}{\eps}\right) =
    T^2\cdot \poly\left(\frac{\log n}{\lambda\eps}\right).
  \]
  The claim of the lemma now follows because by \Cref{lemma:collections}, we have $T=2^{O(1/\eps)}\cdot \ln n$ and because by \Cref{lemma:exOfAug}, we can choose $\lambda = \Theta(\eps/T)$.
\end{proof}

\Cref{lemma:RandApprox} and \Cref{lemma:DetApprox} together complete the proof of \Cref{thm:approxMWM}.


\section{Bipartite Vertex Cover Lower Bound}
\label{sec:lower}

We next proof \Cref{thm:lowerbound}, i.e., we prove that even in bipartite graph of maximum degree $3$, there exists a constant $\eps_0>0$ such that for $\eps\in(0,\eps_0]$, computing a $(1+\eps)$-approximate vertex cover requires $\Omega\big(\frac{\log n}{\eps}\big)$ rounds in the \LOCAL model.

\begin{proof}[\textsf{\textbf{Proof of \Cref{thm:lowerbound}}}]
  in \cite{goeoes14_DISTCOMP}, G\"o\"os and Suomela showed that there exists a bipartite graph $G=(V_G,E_G)$ with maximum degree $3$ and a constant $\eps_0>0$ such that no randomized distributed algorithm with running time $o(\log n)$ can find a $(1+\eps_0)$-approximate vertex cover on $G$. To extend this proof to smaller approximation ratios, we proceed as follows. Given a positive integer parameter $k$, we construct a new lower bound graph $H$ as follows. Graph $H$ is obtained from graph $G$ by replacing every edge $e$ of $G$ by a path $P_e$ of length $2k+1$.

  Assume that we are given a vertex cover $S_H$ of graph $H$. We first describe a method to transform the vertex cover $S_H$ into a vertex cover $S_H'$ of $H$ such that $|S_H'|\leq |S_H|$ and such that $S_H'$ has the following form. For each edge $e$ of $G$, $S_H'$ contains exactly $k$ of the inner nodes of the path $P_e$ in $H$ and it contains at least one of the end nodes of $P_e$. The transformation is done independently for each path $P_e$ as follows. Let $P_e=(v_0,v_1,\dots,v_{2k+1})$ be the path that replaces edge $e=\set{v_0,v_{2k+1}}$ in $G$. If $S_H\cap \set{v_0,v_{2k+1}}\neq \emptyset$, we add the nodes $S_H\cap \set{v_0,v_{2k+1}}$ also to $S_H'$. Otherwise, if $S_H\cap \set{v_0,v_{2k+1}}= \emptyset$, we arbitrarily add either $v_0$ or $v_{2k+1}$ to $S_H'$. Further $S_H'$ contains exactly $k$ of the inner nodes $v_1,\dots,v_{2k}$ of $P_e$ in such a way that every edge of $P_e$ is covered by some node in $S_H'$. If $v_0\in S_H'$, we can add all nodes $v_{2i}$ for $i\in\set{1,\dots,k}$ and otherwise we can add all nodes $v_{2i-1}$ for $i\in\set{1,\dots,k}$. Clearly $S_H'$ is a vertex cover of $H$. To see that $|S_H'|\leq |S_H|$, observe that in order to cover all $2k+1$ edges of $P_e$, every vertex cover of $H$ must contain at least $k+1$ of the nodes of $P_e$ and it also must contain at least $k$ of the inner nodes of $P_e$. If $S_H\cap \set{v_0,v_{2k+1}}\neq \emptyset$, $S_H'$ only differs in terms of the inner nodes of $P_e$ from $S_H$ 
   and we know that also $S_H$ must contain at least $k$ inner nodes of $P_e$. If $S_H\cap \set{v_0,v_{2k+1}}=\emptyset$, we add either $v_0$ or $v_{2k+1}$ to $S_H'$. However in this case, $S_H$ contains at least $k+1$ inner nodes of $P_e$ and $S_H'$ only contains $k$ inner nodes of $P_e$. The transformation from $S_H$ to $S_H'$ can be done independently for each of the paths $P_e$ of length $2k+1$ and it can therefore be done in $O(k)$ rounds in the \LOCAL model.

    Let $e_G$ be the number of edges of $G$ and let $s_G$ be the size
    of an optimal vertex cover of $G$. Because the maximum degree of
    $G$ is $3$ and we have an optimal vertex cover of $G$, we get that
    $e_G=c\cdot s_G$ for some constant $c \leq 3$. By the observation
    above, any vertex cover $S_H$ of $H$ can be transformed into
    an equally good vertex cover $S_H'$ with a nice structure. Note
    that $S_H'$ consists of exactly $k$ inner nodes of each path $P_e$
    for $e\in E_G$ and it consists of at least one of the end nodes of
    each such path (i.e., of a vertex cover of $G$). We therefore
    obtain that there is an optimal vertex cover of $H$ that consists
    of an optimal vertex cover of $G$ and of $k$ inner
    nodes of each $(2k+1)$-hop path $P_e$ replacing an edge $e$ of $G$. The size $s_H$
    of an optimal vertex cover of $H$ is therefore exactly
    $s_H=s_G + k\cdot e_G=(1+ck)\cdot s_G$.

    Assume now that we have a $T$-round algorithm to compute a
    $(1+\eps)$-approximate vertex cover $S_H$ on graph $H$ for some
    $\eps\leq \eps_0/(1+3k)\leq \eps_0/(1+ck)$. By the above
    observation, in $O(k)$ rounds, we can transform this vertex cover
    into a vertex cover $S_H'$, which contains
    $k\cdot e_G=ck\cdot s_G$ inner path nodes and at least one of the
    end nodes of each $(2k+1)$-hop path replacing an edge of $G$ in
    $H$. The vertex cover $S_H'$ of $H$ therefore induces a vertex
    cover of $G$ of size $|S_H'|-k\cdot e_G=|S_H'|-ck\cdot
    s_G$. Because we assumed that $\eps\leq \eps_0/(1+ck)$, we have
    $|S_H'| \leq (1+\eps)\cdot (1+ck)\cdot s_G\leq (1+\eps_0)s_G +
    ck\cdot s_G$. The vertex cover $S_H'$ of $H$ therefore induces a
    $(1+\eps_0)$-approximate vertex cover of $G$. We next show that this
    implies an $O(1+ T/k)$-round algorithm to compute a
    $(1+\eps_0)$-approximate vertex cover of $G$. Assume that we want
    to compute a vertex cover of $G$. To do this, the nodes of $G$ can
    simulate graph $H$ by adding $2k$ virtual nodes on each edge of
    $G$. An $R$-round algorithm on $H$ can then be run in $O(\lceil R/k\rceil)=O(1+R/k)$
    rounds on $G$. We can therefore compute the vertex cover $S_H'$ on
    the virtual graph $H$ in $O(1+(T+k)/k)=O(1+T/k)$ rounds on
    $G$. Since $k=\Theta(1/\eps)$, the lower bound of
    \cite{goeoes14_DISTCOMP} implies a lower bound of
    $\Omega\big(\frac{\log n}{\eps}\big)$ on the time $T$ for
    computing a $(1+\eps)$-approximate vertex cover on $H$. Finally
    note that since $G$ is bipartite, then $H$ is also bipartite and
    clearly if the maximum degree of $G$ is $3$, then the maximum
    degree of $H$ is also $3$.
\end{proof}



\bibliography{references}

\appendix

\section{Basic Tools}
\label{sec:tools}

\subsection{Chernoff Bound}
We use the following Chernoff bound.
\begin{theorem}\label{thm:chernoff_bound}
  Let $X_1, \hdots, X_k$ be independent random variables taking values in $[0, Q]$ for some $Q>0$. Let $X := \sum_{i=1}^k X_i$ and let $\mu$ be such that $\E[X]\geq \mu$. Then,
  for any $\delta \in [0, 1]$, it holds $\Pr\left(X \leq (1-\delta)\mu \right) \leq e^{-\frac{\delta^2 \mu}{2Q}}$.
\end{theorem}

Note that in the standard form, the above Chernoff bound is stated for independent random variables in $\set{0,1}$. A standard convexity argument allows to obtain the same bound for independent random variables in the interval $[0,1]$. The generalization to random variables from the interval $[0,Q]$ is then obtained by scaling.

\subsection{Low-Diameter Clustering}
\label{sec:clustering}

As discussed in \Cref{sec:diameter_reduction}, we need an efficient
deterministic clustering algorithm so that afterwards, we can
concentrate on approximating MWVC and MWM in graphs of polylogarithmic
diameter. The following theorem is a slight generalization of Lemma 13
in \cite{opodis20_MVC}, which by itself is a relatively
straightforward adaptation of one phase of the network decomposition
algorithm of Rozho\v{n} and Ghaffari~\cite{polylogdecomp}.

\begin{theorem}\label{thm:det_clustering}
  Let $G=(V,E,w)$ be a weighted $n$-node graph with node and edge
  weights $w(v),w(e)\in [1,\poly(n)]$, let $h$ be a positive
  integer. For every $\eta\in (0,1]$, there is a deterministic
  \CONGEST algorithm to compute a $(1-\eta)$-dense, $h$-hop separated,
  and
  $\big(O(\log n), O\big(\frac{h\cdot \log^3
    n}{\eta}\big)\big)$-routable clustering of $G$ in
  $\poly\big(\frac{h\cdot \log n}{\eta}\big)$ rounds.
\end{theorem}
\begin{proof} 
  Lemma 13 in \cite{opodis20_MVC} is only stated for edge weights
  $\in\set{0,1}$ and no node weights. However, the proof of Lemma 13 
  explicitly states that the lemma also holds for weights that are
  polynomially bounded in $n$. Further, in the proof of Lemma 13, in
  the first step, the edge weights are transformed into node
  weights. The statement of the lemma therefore directly also applies to node
  weights. Finally, in Lemma 13, the parameter $h$ is set to
  $h=3$. The proof guarantees $3$-hop separation by applying Theorem
  2.12 of \cite{polylogdecomp} to $G^k$ for $k=2$. If the theorem is
  applied to $G^k$ for $k=\max\set{1,h-1}$, Theorem 2.12 of
  \cite{polylogdecomp} directly implies that the computed clusters are
  $(k+1)$-hop separated.  
\end{proof}

\noindent\textbf{Remarks:} In the above proof, we do not optimize the
$\log n$-factors and the $1/\eta$-factors in the time complexity. We
note that the round complexity of the algorithm of
\cite{polylogdecomp} has been slightly improved in
\cite{improved_decomp}. Further, in \cite{strongdiameterdecomp}, Chang
and Ghaffari give a \CONGEST algorithm to compute a strong diameter
network deceomposition. We think that it should be possible to adapt
their algorithm to obtain a deterministic $\poly\log n$-time \CONGEST
algorithm to compute an $h$-hop separated $(1-\eta)$-dense clustering
into clusters of strong diameter
$O\big(\frac{h\cdot \log^2 n}{\eta}\big)$. We further note that by
using a variant of the randomized clustering algorithm of
\cite{MPX13}, one gets a randomized
$O\big(\frac{h\cdot \log n}{\eta}\big)$-round \CONGEST algorithm to
compute a clustering into clusters of diameter
$O\big(\frac{h\cdot \log n}{\eta}\big)$ such that the clustering is
$(1-\eta)$-dense in expectation.

\subsection{Fractional Approximation Algorithm}
\label{sec:fractionalapprox}

In all our upper bounds, we need an efficient deterministic
distributed approximation scheme for the fractional variants of the
MWVC and the MWM problem. In \cite{ahmadi18}, Ahmadi, Kuhn, and Oshman
showed that for every instance of the MWM problem with weights in the
range $[1,W]$ and for every $\eps\in (0,1]$, it is possible to
deterministically compute a $(1-\eps)$-approximate fractional solution
in time $O(\log(\Delta W)/\eps^2)$ in the \CONGEST model. In our
algorithm to solve the MWVC problem, we need a variant of this
algorithm, which works for the (unweighted) fractional $w$-matching
problem. The following theorem shows that this can be done with the
same asymptotic cost that we have for the fractional MWM problem.

\begin{theorem}\label{thm:fractional_wmatching}
  Let $G=(V,E,w)$ be an undirected $n$-node graph with integer node weights
  $w(v)\in \set{1,\dots,W}$. Then, for every $\eps\in(0,1]$, there is a deterministic
  $O(\log(\Delta W)/\eps^2)$-round \CONGEST algorithm to compute a
  $(1-\eps)$-approximate solution to the fractional $w$-matching
  problem in $G$ and a $(1+\eps)$-approximate solution to the minimum
  fractional weighted vertex cover problem.
\end{theorem}
\begin{proof}
  The theorem could be proven for arbitrary weights in the range
  $[1,W]$ by adapting the algorithm of \cite{ahmadi18}. We here give a
  generic reduction, which works for integer weights and which allows
  to use the result of \cite{ahmadi18} (almost) in a blackbox manner.

  We define an unweighted graph $G'=(V',E')$ as follows. For every
  node $v\in V$, $V'$ contains $w(v)$ nodes
  $(v,1),\dots,(v,w(v))$. Further, for every edge $\set{u,v}\in E$, we
  add a complete bipartite graph between the corresponding nodes to
  $G'$, that is, $E'$ contains all edges $\set{(u,i),(v,j)}$ for
  $i\in\set{1,\dots,w(v)}$ and $j\in\set{1,\dots,w(u)}$. We then apply
  the unweighted fractional matching algorithm of \cite{ahmadi18} to
  compute a $(1-\eps)$-approximate fractional matching of $G'$. The
  maximum degree of any node in $G'$ is at most $\Delta\cdot W$ and
  when running the algorithm in the \CONGEST model on $G'$, the round
  complexity of the algorithm is therefore $O(\log(\Delta W)/\eps^2)$
  as claimed (cf.\ Theorem 2 in \cite{ahmadi18}).

  For an edge $e\in E'$ of $G'$, assume that $z_{e}$ is the fractional
  matching value of $e$ in the computed fractional matching of
  $G'$. We can transform the fractional matching of $G'$ into a
  fractional $w$-matching of $G$ as follows. For each edge
  $\set{u,v}\in E$ of $G$, we define
  $y_{\set{u,v}} :=
  \sum_{i=1}^{w(u)}\sum_{j=1}^{w(v)}z_{\set{(u,i),(v,j)}}$. Note that
  we obtain a valid fractional $w$-matching because for every node
  $u\in V$ of $G$, the sum of the fractional values of its edges is at
  most equal to number of copies of $u$ in $G'$, which is equal to
  $w(u)$. In the other direction, given a fractional $w$-matching
  $y_e$ of $G$, we can compute a fractional matching $z_{e'}$ of $G'$
  of the same size in the following way. For each edge
  $\set{u,v}\in E$ of $G$, we assign
  $z_{\set{(u,i),(v,j)}}:= y_{\set{u,v}}/(w(u)\cdot w(v))$. The size
  of a maximum fractional $w$-matching on $G$ is therefore equal to
  the size of a maximum fractional matching on $G'$ and given a
  $(1-\eps)$-approximation of maximum fractional matching on $G'$, we
  therefore also obtain a $(1-\eps)$-approximation of maximum
  fractional $w$-matching on $G$.

  It remains to show that we can efficiently run the fractional
  matching algorithm of \cite{ahmadi18} in the \CONGEST model on
  $G$. Since every node of $G$ has potentially a large number of
  copies in $G'$, it is not true that any \CONGEST algorithm on $G'$
  can be run efficiently in the \CONGEST model on $G$. However, the
  behavior of the algorithm of \cite{ahmadi18} is independent of the
  node IDs and since the algorithm is deterministic, all copies of a
  node $u\in V$ are symmetric and therefore behave in exactly the same
  way. Each node of $G$ can therefore simulate all its copies in $G'$
  at no additional cost.

  We note that the same reduction has also been used in
  \cite{GrandoniKP08}, where a maximal matching of $G'$ is used to
  compute a $2$-approximate weighted vertex cover of $G$.
\end{proof}

\subsection{Counting Paths in Layered Graphs}
\label{sec:pathcounting}

As part of our bipartite minimum vertex cover algorithm, we need to
count the number of shortest augmenting paths (w.r.t.\ to a given
fractional $w$-matching) that pass through each node and edge of a
bipartite graph. For this, we slightly generalize an algorithm that
has been introduced in \cite{lotker15} and refined in
\cite{yehuda17}. The following lemma provides a generic version of
this algorithm, which works for general layered bipartite graphs.

\begin{lemma}\label{lemma:pathcounting}
  Let $H=(V,E)$ be a bipartite $n$-node graph, where the nodes are partitioned into $L$
  layers $V_1,\dots,V_L$ such that $E$ only contains edges between
  adjacent layers $V_i$ and $V_{i+1}$ for $i\in\set{1,\dots,L-1}$. A
  top-down path $P$ is a path of length $L-1$ that consists of exactly
  $1$ node from each layer $V_i$. If the partition into layers is
  known, there is a deterministic $O(L^2)$-round \CONGEST algorithm that for every
  $v\in V$ and every $e\in E$ computes the number of top-down paths
  passing through $v$ and $e$.
\end{lemma}
\begin{proof} 
  For every $i\in \set{1,\dots,L}$ and every $v\in V_i$, we first
  compute the number of paths $\alpha(v)$ of length $i-1$ that start at a node in
  $V_1$ and pass through the layers $V_1,\dots,V_i$. The number can be
  computed inductively as follows. For each $v\in V_1$, we have
  $\alpha(v)=1$ and for each $v\in V_i$ for $i>1$, we have
  $\alpha(v)=\sum_{u\in N(v)\cap V_{i-1}}\alpha(u)$. Similarly, for
  every $i\in \set{1,\dots,L}$ and every $v\in V_i$, we can also compute
  the number of paths $\beta(v)$ of length $L-i$ that start at a node
  in $V_L$ and pass through the layers $V_L,V_{L-1},\dots,V_i$. Note
  that we have $\alpha(v),\beta(v)\leq n^L$ for all nodes $v$ and we
  can therefore compute all values $\alpha(v)$ and $\beta(v)$ by a
  distributed algorithm in $L$ rounds with messages of size $O(L\log
  n)$. The values can therefore be computed in the \CONGEST model in
  time $O(L^2)$. The number of top-down path passing through a node
  $v$ can now be computed as $\alpha(v)\cdot \beta(v)$ and for an edge
  $\set{u,v}\in E$ with $u\in V_i$ and $v\in V_{i+1}$ for some $i\in
  \set{1,\dots,L-1}$, the number of top-down paths passing through
  edge $\set{u,v}$ is equal to $\alpha(u)\cdot\beta(v)$.
\end{proof}


\end{document}